\documentclass[12pt]{article}

\usepackage[a4paper]{geometry}
\usepackage{amssymb}
\usepackage{amsmath}
\usepackage{amsthm}

\theoremstyle{plain}
\newtheorem{theorem}{Theorem}[section]
\newtheorem{proposition}[theorem]{Proposition}
\newtheorem{lemma}[theorem]{Lemma}
\newtheorem{corollary}[theorem]{Corollary}

\theoremstyle{definition}

\newtheorem{assumption}[theorem]{Assumption}

\newtheorem{remark}[theorem]{Remark}
\newtheorem*{remark*}{Remark}
\newtheorem{remarks}[theorem]{Remarks}
\newtheorem*{remarks*}{Remarks}

\numberwithin{equation}{section}

\title{Elementary commutator method for the Dirac equation with long-range perturbations}
\author{
Shinichi {\scshape Arita}\footnote{Graduate School of Mathematical Sciences, 
The University of Tokyo, 3-8-1 Komaba, Meguro-ku, Tokyo 153-8914, Japan.
E-mail: \texttt{shinichi0620a@g.ecc.u-tokyo.ac.jp}. 
}
\ \& 
Kenichi {\scshape Ito}\footnote{Department of Mathematics, Graduate School of Science, Kobe University,
1-1, Rokkodai, Nada-ku, Kobe 657-8501, Japan.
E-mail: \texttt{ito-ken@math.kobe-u.ac.jp}. 
}
}
\date{}

\begin{document}

\allowdisplaybreaks

\maketitle

\begin{abstract}
We present direct and elementary commutator techniques for the Dirac equation
with long-range electromagnetic and mass perturbations. 
The main results are absence of generalized eigenfunctions and locally uniform resolvent estimates,
both in terms of the optimal Besov-type spaces.  
With an additional massless assumption, we also obtain an algebraic radiation condition of projection type. 
For their proofs, following the scheme of Ito--Skibsted (2020), we adopt, 
along with various weight functions, 
the generator of radial translations as conjugate operator,
and avoid any of advanced functional analysis, pseudodifferential calculus, or even reduction to the Schr\"odinger equation. 
The results of the paper would serve as a foundation for the stationary scattering theory of the Dirac operator. 
\end{abstract}



\maketitle

\tableofcontents

\section{Settings and results}

\subsection{Settings}

\subsubsection{Introduction}

In this paper we investigate the Dirac equation 
on the Euclidean space $\mathbb R^{d}$ of dimension $d\in \mathbb N=\{1,2,\ldots\}$.
We write it in the form of evolution equation as 
\[
 \partial_t\psi
=
-\mathrm iH\psi
\]
with 
\[
H=\alpha_jp_j+q,
\quad p_j=-\mathrm i\partial_j\ \text{for }j=1,\dots,d.
\]
Throughout the paper we adopt Einstein's summation convention 
however without tensorial superscripts. 
Here, for some $n\in\mathbb N$, $\psi$ is 
an unknown $\mathbb C^n$-valued function,
and $q$ is a given $n\times n$ Hermitian matrix-valued function, 
including simultaneously an electromagnetic potential and a mass term. 
We assume that $q$ is time-independent and long-range as will be formulated below. 
In addition, $\alpha_1,\dots ,\alpha_d$ are $n\times n$ Hermitian matrices satisfying 
the anti-commutation relations 
\begin{equation}
\{\alpha_j,\alpha_k\}=2\delta_{jk}I_n
\ \ \text{for }j,k=1,\dots,d
,
\label{25071914}
\end{equation}
where $I_n$ is the identity matrix, which is omitted in the following. 
It is well-known that such matrices exist if $n=2^{[d/2]}$, 
or if $n=2^{[(d+1)/2]}$ when $\beta$ satisfying \eqref{251025} below is added.
However, our arguments are independent of $n$ or of explicit representations of $\alpha_1,\dots,\alpha_d$, 
and dependent only on the above anti-commutation relations.

The purpose of the paper is to introduce a new commutator scheme to the analysis of the Dirac operator, 
following a paper by Ito--Skibsted~\cite{MR4062329}, which originally applied it to the Schr\"odinger operator. 
The scheme is quite elementary. 
The conjugate operator is simply the generator of radial translations, 
and the main tools are the product rule of derivatives and the Cauchy--Schwartz inequality,
combined with exquisite choice of weight functions.
We do not need energy cut-offs from functional analysis or pseudodifferential calculus, 
or even reduction to the Schr\"odinger operator. 
See the papers~\cite{MR1743451,MR1860844,MR2134232,MR2134850,MR3816195} for the Mourre theory with more technical conjugate operators. 
Then the difficulty of the paper is that the commutator is merely a first-order differential operator, for which we cannot generally expect a sign. 
However, we can extract (weighted) positivity from it, which is the most technical novelty of the paper.
We remark that a commutator argument of this type may be called 
the \textit{virial theorem}~\cite{MR400989,MR1832264,MR1927172,
MR1957634, 
MR2009334,
MR4679382}
or 
the \textit{method of multiplier}~\cite{MR4156219}, however we keep to use the terminology ``commutator''. 

The first result of the paper is \textit{Rellich's theorem}, or absence of generalized eigenfunctions 
in an optimally weighted Besov-type space, called the \textit{Agmon-H\"ormander space}. 
Although that of $L^2$-eigenfunctions has long been discussed in various settings~\cite{
MR400989,
MR928641,
MR1832264,
MR1927172,
MR1957634, 
MR2009334, 
MR4156219,
MR4316739,
MR4679382}, 
only Vogelsang~\cite{MR923045} seems to have considered it in 
a wider weighted space $L^2_{-1/2}(\mathbb R^d)$ for $d=3$. 
Hence we might for the first time extend it to the optimally largest space,
which properly contains $L^2_{-1/2}(\mathbb R^d)$. 
Note that, to this end, the choice of our conjugate operator is essential, 
while the generator of dilations, adopted in the ordinary Mourre theory, cannot go beyond $L^2(\mathbb R^d)$. 
The second result is the local \textit{LAP bounds}, or locally uniform resolvent estimates 
as operators between certain Besov-type spaces, which we also call the Agmon-H\"ormander spaces. 
We obtain them with a long-range electric and mass perturbation in a general dimension. 
This subject also has a long history \cite{MR320547,MR923045,MR1162140,MR1846785}, 
and recent interest is more in high energy, threshold or global estimates 
\cite{MR1249220,
MR1937841,
MR1743451,
MR2718928,
MR3526232,
MR3793450,
MR3816195,
MR3933410, 
MR4601596}. 
Our result is still local, however 
adoption of the Agmon-H\"ormander spaces would be new for the Dirac operator of such generality.  
Finally, with an additional massless assumption, we verify \textit{Sommerfeld's uniqueness}
for a solution to the associated Helmholtz equation. 
For that we present two types of radiation conditions, the \textit{analytic} and the \textit{algebraic} ones. 
The former is a natural analogue from the Schr\"odinger theory, 
extending a result by Pladdy--Sait{\=o}--Umeda~\cite{MR1625972} to long-range perturbations in the massless case. 
On the other hand, as for the latter, 
we construct an explicit orthogonal projections onto the outgoing and the incoming subspaces of $\mathbb C^n$. 
This extends a part of results by Kravchenko--Castillo P.~\cite{MR1949503} and 
Marmolejo-Olea--P\'erez-Esteva~\cite{MR3072343}. 
Let us remark that also Vogelsang~\cite{MR923045} seems to have discussed it in an indirect manner. 

Whereas these results have their own value, 
they would also serve as fundamental ingredients for the stationary scattering theory, 
as in the Schr\"odinger case~\cite{MR4398220,MR4874313}. 
So far, while the time-dependent scattering theory for the Dirac operator has been intensively studied 
\cite{MR711298,MR921790,MR1743451,MR2134232,MR2134850,MR4072582}, 
the stationary theory seems to have been less developed, and done mainly for 
the short-range perturbations~\cite{MR1162140,MR407471}. 
Though G\^atel--Yafaev~\cite{MR1846785} took stationary approach to the long-range case, 
their wave operators are formulated in a time-dependent manner. 
Hopefully, we could discuss elsewhere the \textit{purely} stationary scattering theory 
for long-range perturbations based on the results of the paper.

\subsubsection{Long-range perturbations}

To state precise assumptions on $q$, let us introduce a \textit{modified radius} and 
the associated radial derivative. 
Choose $\chi\in C^\infty(\mathbb{R})$ such that 
\begin{align}
\chi(t)
=\left\{\begin{array}{ll}
1 &\mbox{ for } t \le 1, \\
0 &\mbox{ for } t \ge 2,
\end{array}
\right.
\quad
\chi'\le 0,
\label{eq:14.1.7.23.24}
\end{align}
and define $f\in C^\infty(\mathbb R^d)$ as 
\[
f(x)=
\chi(|x|)+|x|(1-\chi(|x|))\ \ \text{for }x\in\mathbb R^d.
\]
Obviously, $f$ coincides with $r=|x|$ for $|x|\ge 2$, but the values for $|x|<2$ are modified so that 
$f$ is uniformly positive on the whole $\mathbb R^d$. 
Thus, in particular, we may use $f$ as a weight function replacing the standard one $\langle x\rangle=(1+x^2)^{1/2}$. 
In addition, we define 
\begin{equation}
\partial_f=(\partial_jf)\partial_j
,
\label{25710}
\end{equation}
which is non-singular and globally defined on $\mathbb R^d$. 
We denote the set of all $n\times n$ Hermitian matrices by $\mathbb C^{n\times n}_{\mathrm{Her}}$,
and its norm, arbitrarily fixed one, by $|\cdot|$.

\begin{assumption}\label{2571017}
The perturbation $q$ splits as 
\begin{align*}
q=q_0+q_1+q_2;\quad q_0,q_1\in C^1(\mathbb R^d;\mathbb C^{n\times n}_{\mathrm{Her}}),\ \  
q_2\in L^\infty(\mathbb R^d;\mathbb C^{n\times n}_{\mathrm{Her}}),
\end{align*}
such that the identities 
\begin{equation}
\widetilde q_0:=\alpha_1q_0\alpha_1=\alpha_2q_0\alpha_2=\dots=\alpha_dq_0\alpha_d 
\label{250719}
\end{equation}
hold on $\mathbb R^d$, and that for some $\rho\in(0,1]$ and $C>0$
\begin{align}\label{eq:60}
|q_1|\le Cf^{-(1+\rho)/2},\quad 
|\partial_1 q_0|+\dots+|\partial_d q_0|+|\partial_f q_1|+|q_2|\le Cf^{-1-\rho}
\end{align}
hold on $\mathbb R^d$. 
\end{assumption}
\begin{remarks}\label{251003}
\begin{enumerate}
\item 
The long-range part $q_0$ may be considered as a sum of electric potential and mass
possibly with non-zero limit at spatial infinity. 
To see this, consider the \textit{ordinary} Dirac equation with electromagnetic potential 
$A=(V,-A_1,\dots,-A_d)$ and mass $m$: 
\[
\partial_t\psi
=
-\mathrm i(\alpha_jp_j+\beta m+V-\alpha_jA_j)\psi
,
\]
where $\beta\in \mathbb C^{n\times n}_{\mathrm{Her}}$ satisfies the anti-commutation relations
\begin{equation}
\{\alpha_j,\beta\}=0\ \  \text{for }j=1,\dots,d
,\quad 
\beta^2=1
.
\label{251025}
\end{equation}
Our model covers this case with $q=\beta m+V-\alpha_jA_j$,
and $q_0=V+\beta m$ in fact satisfies \eqref{250719} with $\widetilde q_0=V-\beta m$.

\item 
The identity \eqref{250719} can be relaxed to admit short-range errors.

\item
We can further add local singularities, as long as $H$ with domain $C^\infty_{\mathrm c}(\mathbb R^d;\mathbb C^n)$ 
is essentially self-adjoint on $\mathcal H=L^2(\mathbb R^d;\mathbb C^n)$, and the unique continuation property is available. 
Note that the essential self-adjointness is required only for justification of doing and undoing commutators,  
as extensions from $C^\infty_{\mathrm c}(\mathbb R^d;\mathbb C^n)$. 
Such justification is not discussed in this paper, and we refer to the previous work \cite{MR4062329}. 
\end{enumerate}
\end{remarks}

\subsubsection{Agmon--H\"ormander spaces}

All of  our main results are stated in terms of the \textit{Agmon--H\"ormander spaces}, 
or the \textit{Besov spaces} associated with the multiplication operator $f$ on $\mathcal H$.
See also a paper by Jensen--Perry~\cite{MR789384}. 
Here let us recall them. 
Let $F(S)$ be the characteristic function of a given  subset $S\subset \mathbb R^d$,
and set
\begin{equation}
F_\nu=F\bigl(\{ x\in \mathbb R^d;\ 2^{\nu}\le f<2^{\nu+1}\} \bigr)
\quad \text{for }\nu\in\mathbb N_0=\{0\}\cup \mathbb N.
\label{250729}
\end{equation}
Then we define function spaces $\mathcal B$, $\mathcal B^*$ and $\mathcal B^*_0$ as 
\begin{align*}
\mathcal B&=\bigl\{\psi\in L^2_{\mathrm{loc}}(\mathbb R^d;\mathbb C^n);\ \|\psi\|_{\mathcal B}<\infty\bigr\}
,\quad 
\|\psi\|_{\mathcal B}=\sum_{\nu\in\mathbb N_0} 2^{\nu/2}\|F_\nu\psi\|_{{\mathcal H}}
,\\
\mathcal B^*&=\bigl\{\psi\in L^2_{\mathrm{loc}}(\mathbb R^d;\mathbb C^n);\ \|\psi\|_{\mathcal B^*}<\infty\bigr\}
,\quad 
\|\psi\|_{\mathcal B^*}=\sup_{\nu\in\mathbb N_0}2^{-\nu/2}\|F_\nu\psi\|_{{\mathcal H}},
\\
\mathcal B^*_0
&=
\Bigl\{\psi\in \mathcal B^*;\ \lim_{\nu\to\infty}2^{-\nu/2}\|F_\nu\psi\|_{{\mathcal H}}=0\Bigr\},
\end{align*}
respectively. 
Note that, if we denote the \textit{weighted $L^2$ spaces} by 
\[
L_s^2=f^{-s}L^2(\mathbb R^d;\mathbb C^n)\ \ \text{for }s\in\mathbb R 
,
\]
then for any $s>1/2$ we have strict inclusions 
\begin{align*}
 L^2_s\subsetneq \mathcal B\subsetneq L^2_{1/2}
\subsetneq \mathcal H
\subsetneq L^2_{-1/2}\subsetneq \mathcal B^*_0\subsetneq \mathcal B^*\subsetneq L^2_{-s}.
\end{align*}

\subsection{Main results}

\subsubsection{Rellich's theorem}

Our first main theorem is \textit{Rellich's theorem}. 
We set 
\begin{align*}
m_+&=\liminf_{R\to\infty}\{\lambda\in\mathbb R;\ \lambda \ge q_0(x)\ \text{for all }|x|\ge R\} 
,\\
m_-&=\limsup_{R\to\infty}\{\lambda\in\mathbb R;\ \lambda \le q_0(x)\ \text{for all }|x|\ge R\} 
,
\end{align*}
where the inequalities involving $\lambda$ and $q_0(x)$ are those for matrices, 
or quadratic forms on $\mathbb C^n$. 

\begin{theorem}\label{250717}
Suppose Assumption~\ref{2571017}.
If $\phi\in  \mathcal B^*_0$ and $\lambda\in \mathbb R\setminus [m_-,m_+]$ satisfy
\begin{align}
(H-\lambda)\phi=0\ \ \text{in the distributional sense,}
\label{25071112}
\end{align}
then $\phi=0$ as a function on $\mathbb R^d$.
In particular, 
the self-adjoint realization of $H$ has no eigenvalues outside $[m_-,m_+]$,
i.e., $\sigma_{\mathrm{pp}}(H)\setminus [m_-,m_+]=\emptyset$.
\end{theorem}

\begin{remarks}
\begin{enumerate}
\item
The function space $\mathcal B^*_0$ is optimal since 
we can in fact construct a non-trivial solution to \eqref{25071112} in $\mathcal B^*$
by using a WKB-type approximation and Corollary~\ref{cor:Sommerfeld-unique-result} below, 
see, e.g., \cite[Proposition~4.18]{MR4874313}. 
So far the largest space where the generalized eigenfunctions are absent 
is $L^2_{-1/2}(\mathbb R^d)$ for $d=3$ due to Vogelsang~\cite{MR923045}, as far as the authors are aware of.

\item
As remarked in a useful survey by Ito--Yamada \cite{MR4098339}, 
if $q$ is smooth and purely electric, we can reduce it to the Schr\"odinger operator. 
In fact, apply $\alpha^jp_j$ to \eqref{25071112} 
and use the identity \eqref{25071112} itself, and then we can deduce 
\begin{equation}
\bigl(p_jp_j+[\alpha_jp_j,q]-(\lambda-q)^2\bigr)\phi=0
.
\label{25082523}
\end{equation}
However, in general case it is not straightforward how to control $[\alpha_jp_j,q]$ since $\alpha_j$ and $q$ do not commute. 
We shall not go through this expression, but directly study the identity \eqref{25071112}.
\end{enumerate}
\end{remarks}

\subsubsection{LAP bounds}

We next discuss the \textit{LAP} (\textit{Limiting Absorption Principle}) \textit{bounds}
for the resolvent
\begin{align*}
R(z)=(H-z)^{-1}\in\mathcal L(\mathcal H)\ \ \text{for }z\in\mathbb C\setminus \mathbb R
.
\end{align*}
For any compact interval $I\subset \mathbb R\setminus [m_-,m_+]$ let us denote
\begin{align}
I_\pm=\bigl\{z=\lambda+\mathrm i\Gamma\in \mathbb C;\ \lambda\in I,\ 0<\pm \sigma \Gamma\le 1\bigr\}
,
\label{25100422}
\end{align}
respectively, where $\sigma=+$ if $I\subset (m_+,\infty)$, and $\sigma=-$ if $I\subset (-\infty,m_-)$.

\begin{theorem}\label{thm:12.7.2.7.9}
Suppose Assumption~\ref{2571017}, 
and let $I\subset \mathbb R\setminus[m_-,m_+]$ be a compact interval.
Then there exists $C>0$ such that 
for any $\phi=R(z)\psi$ with $z\in I_\pm$ and $\psi\in \mathcal B$
\begin{equation*}
\|\phi\|_{\mathcal B^*}
+\|p_1\phi\|_{\mathcal B^*}
+\dots
+\|p_d\phi\|_{\mathcal B^*}
\le C\|\psi\|_{\mathcal B}.
\end{equation*}
In particular, 
the self-adjoint realization of $H$ has no singular continuous spectrum outside $[m_-,m_+]$,
i.e., $\sigma_{\mathrm{sc}}(H)\setminus [m_-,m_+]=\emptyset$.
\end{theorem}

\begin{remark}
For the proof of the LAP bounds it seems to have been typical to somehow reduce it to the (free) Schr\"odinger operator.
In fact, similarly to \eqref{25082523}, we can deduce from the identity 
$(H-z)\phi=\psi$ that 
\[
\bigl(p_jp_j+[\alpha_jp_j,q]-(z-q)^2\bigr)\phi
=
(\alpha_jp_j-q+z)\psi
.
\]
However, again, we provide a direct proof, 
employing an elementary commutator method very similar to that for Rellich's theorem.
Rellich's theorem itself also plays an essential role in contradiction argument in the proof of the LAP bounds.   
\end{remark}

\subsubsection{Radiation condition bounds and applications}

For further arguments we require the following massless assumption.

\begin{assumption}\label{2571017b}
In addition to Assumption~\ref{2571017}, there exists $C'\ge 0$ such that 
\[
|(q_0+\widetilde q_0)q_1-q_1(q_0+\widetilde q_0)|\le C'f^{-1-\rho}
,\quad 
|q_0-\widetilde q_0|\le C'f^{-\rho}
\]
hold on $\mathbb R^d$.
\end{assumption}
\begin{remark}
In terms of the ordinary model given in Remark~\ref{251003}, 
$q_0=V+\beta m$ satisfies the former bound as it is.
However, the latter requires $m$ to be of order $f^{-\rho}$, 
and this is exactly the massless assumption at spatial infinity.
\end{remark}

The following \textit{radiation condition bounds} describe radial oscillatory behavior of the resolvent 
at spatial infinity. Using $\sigma$ defined right after \eqref{25100422}, let us set 
\begin{equation}
p_f=(\partial_jf)p_j
,\quad 
\alpha_f=(\partial_jf)\alpha_j
,\quad 
\pi_{\pm} =\tfrac12(1\pm\sigma\alpha_f),
\label{251004}
\end{equation}
respectively, cf.\ \eqref{25710}. 
Due to \eqref{25071914}, 
$\pi_{\pm}$  for $|x|\ge 2$ are orthogonal projections on $\mathbb C^n$.
They represent the projections 
onto the subspaces of outgoing/incoming components, respectively, as can be seen in the following.

\begin{theorem}\label{25081317}
Suppose Assumption~\ref{2571017b},
let $I\subset \mathbb R\setminus[m_-,m_+]$ be a compact interval, and let $\kappa\in (0,\rho/2)$.
Then there exists $C>0$ such that 
for any $\phi=R(z)\psi$ with $z\in I_\pm$ and $\psi\in L^2_{1/2+\kappa}$
\begin{equation*}
\|\pi_{\mp}\phi\|_{L^2_{-1/2+\kappa}}
+\|(p_f-\alpha_f(z-q_0))\phi\|_{L^2_{-1/2+\kappa}}
\le C\|\psi\|_{L^2_{1/2+\kappa}}
,
\end{equation*}
respectively. 
\end{theorem}
\begin{remarks}
\begin{enumerate}
\item
The first term on the left-hand side corresponds to the \textit{algebraic} radiation condition 
due to Kravchenko--Castillo P.~\cite{MR1949503} and Marmolejo-Olea--P\'erez-Esteva~\cite{MR3072343}. 
While they discussed an exterior problem without electromagnetic potentials, 
we do it for long-range perturbations without obstacles under massless assumption. 
It would not be difficult to extend our result to the exterior case, even with some unbounded obstacles, 
following the setting of Ito--Skibsted~\cite{MR4062329}, but we shall not elaborate it in the present paper. 
On the other hand, the \textit{analytic} radiation condition, like the one by Pladdy--Sait{\=o}--Umeda~\cite{MR1625972},
can be obtained by combining the first and the second terms on the left-hand side,
see Corollary~\ref{cor:Sommerfeld-unique-result}. 
\item
The proof of Theorem~\ref{25081317} is more or less different from Ito--Skibsted~\cite{MR4062329}. 
\end{enumerate}
\end{remarks}

Then let us present applications of Theorems~\ref{250717}, 
\ref{thm:12.7.2.7.9} and \ref{25081317}. 
The first one is the \textit{LAP}, or existence of the limits of resolvent as $I_\pm\ni z\to \lambda\in I$. 
Denote by $\mathcal L(X,Y)$ the space of all bounded operators 
from a Banach space $X$ to another $Y$. 

\begin{corollary}\label{cor:Limiting-Absorption-Principle-Stark}
Suppose Assumption~\ref{2571017b},
let $I\subset\mathbb R\setminus[m_-,m_+]$ be a compact interval,
and let $s\in (1/2,(1+\rho)/2)$ and $\epsilon\in(0,s-1/2)$. 
Then there exists $C>0$ such that for any $z,w\in I_+$ or $z,w\in I_-$ 
\begin{equation}\label{eq:Holder-continuity}
\begin{split}
\|R(z)-R(w)\|_{\mathcal L(L^2_s,L^2_{-s})}
+\sum_{j=1}^d\|p_jR(z)-p_jR(w)\|_{\mathcal L(L^2_s,L^2_{-s})}
\le C|z-w|^\epsilon.
\end{split}
\end{equation}
In particular, $R(z), p_1R(z),\dots,p_dR(z)$ have uniform limits 
as $I_\pm \ni z\to\lambda\in I$ 
in the norm topology of $\mathcal L(L^2_s,L^2_{-s})$. 
In addition, if one denotes 
\begin{equation}\label{eq:uniform-limit-z-to-lambda}
\begin{split}
R_\pm (\lambda) = \lim_{I_\pm\ni z\to \lambda}R(z)
\ \ \text{in }\mathcal L(L^2_s,L^2_{-s}), 
\end{split}
\end{equation}
respectively, 
then $R_\pm (\lambda),
p_1R_\pm (\lambda),\dots,p_dR_\pm (\lambda)$ belong to 
$\mathcal L(\mathcal B,\mathcal B^*)$.
\end{corollary}

Theorem~\ref{25081317} immediately extends to $R_\pm (\lambda)$ as follows.

\begin{corollary}\label{cor:RC-bound-real}
Suppose Assumption~\ref{2571017b},
let $I\subset\mathbb R\setminus[m_-,m_+]$ be a compact interval,
and let $\kappa\in (0,\rho/2)$.
Then there exists $C>0$ such that 
for any $\phi=R_\pm (\lambda)\psi$ 
with $\lambda\in I$ and $\psi\in L^2_{1/2+\kappa}$ 
\begin{equation*}
\|\pi_{\mp}\phi\|_{L^2_{-1/2+\kappa}}
+\|(p_f-\alpha_f(\lambda-q_0))\phi\|_{L^2_{-1/2+\kappa}}
\le C\|\psi\|_{L^2_{1/2+\kappa}}
,
\end{equation*}
respectively. 
\end{corollary}

The following \emph{Sommerfeld's uniqueness}
characterizes $R_\pm (\lambda)$ as solution operators 
to the Helmholtz equation with outgoing or incoming radiation conditions.

\begin{corollary}\label{cor:Sommerfeld-unique-result}
Suppose Assumption~\ref{2571017b}, and let $\lambda\in\mathbb R\setminus[m_-,m_+]$ and $\kappa\in (0,\rho/2)$. 
If $\psi\in L^2_{1/2+\kappa}$, then $\phi:=R_\pm (\lambda)\psi\in\mathcal B^*$, and it satisfies 
\begin{enumerate}
\item\label{item:18122818}
$(H-\lambda)\phi=\psi$ in the distributional sense,
\item\label{item:18122819}
$\pi_{\mp}\phi\in L^2_{-1/2+\kappa}$,
and $(p_f\mp\sigma(\lambda-q_0))\phi\in L^2_{-1/2+\kappa}$, 
\end{enumerate}
respectively. 
Conversely, if $\psi\in f^{-\kappa}\mathcal B$ and $\phi\in f^{\kappa}\mathcal B^*$ satisfy the above \ref{item:18122818} and 
\begin{enumerate}
\setcounter{enumi}{2}
\item\label{item:18122820}
$\pi_{\mp}\phi\in f^{-\kappa}\mathcal B_0^*$, 
or $(p_f\mp\sigma(\lambda-q_0))\phi\in f^{-\kappa}\mathcal B_0^*$, 
\end{enumerate}
then $\phi=R_\pm (\lambda)\psi$, respectively.
\end{corollary}

\section{Proof of Rellich's theorem}

\subsection{Reduction to two propositions} 

In this section we prove Theorem~\ref{250717}, following the scheme of Ito--Skibsted~\cite{MR4062329}. 
We split the proof into two steps. 
Obviously, Theorem~\ref{250717} is a consequence of the following propositions. 
Throughout the section we assume Assumption~\ref{2571017}. 

\begin{proposition}\label{prop:absence-eigenvalues-1b}
Let $\phi\in \mathcal B^*_0$ and $\lambda\in\mathbb R\setminus [m_-,m_+]$ satisfy
\[
(H-\lambda)\phi=0\ \ \text{in the distributional sense.}
\]
Then for any $\kappa>0$ one has $\mathrm e^{\kappa f}\phi\in \mathcal B_0^*$.
\end{proposition}

\begin{proposition}\label{prop:absence-eigenvalues-1bbb}
Let $\phi\in \mathcal B_0^*$ and $\lambda\in\mathbb R\setminus [m_-,m_+]$ satisfy 
\begin{enumerate}
\item
$(H-\lambda)\phi=0$ in the distributional sense, 
\item
$\mathrm e^{\kappa f}\phi\in \mathcal B_0^*$ for any $\kappa\ge 0$.
\end{enumerate}
Then $\phi=0$ as a function on $\mathbb R^d$.
\end{proposition}

We will prove these propositions 
in Sections~\ref{subsec:150909511} and \ref{subsec:150909513}, respectively, 
after short preliminaries in Section~\ref{250718}.

\subsection{Preliminaries}\label{250718}

\subsubsection{Conjugate operator}

Here we present formulas, to be repeatedly referred to in this and later sections, involving a general weight function. 
The proofs of Propositions~\ref{prop:absence-eigenvalues-1b} and \ref{prop:absence-eigenvalues-1bbb} 
depend on a commutator method sharing conjugate operators of the same form, however, with different weight functions. 

Given a weight function $\Theta$, we introduce a radial differential operator 
\begin{equation}
A_\Theta
=2\mathop{\mathrm{Re}}(\Theta p_f)
=
\Theta p_f+p_f^*\Theta
,
\label{260728}
\end{equation}
where $p_f^*$ denotes the adjoint of $p_f$ from \eqref{251004}, unlike the physical convention. 
The function $\Theta$ will be specified soon below, 
but for the moment we only assume the following properties. 
We let $\Theta$ be a smooth function only of $f$ satisfying
\begin{align}
\Theta\ge 0,\quad 
\mathop{\mathrm{supp}}\Theta\subset \{x\in\mathbb R^d;\ |x|\ge 2\},\quad 
|\Theta^{(k)}|\le C_k\ \ \text{for any }k\in\mathbb N_0
,
\label{eq:150921}
\end{align}
with $\Theta^{(k)}$ being the $k$-th derivative of $\Theta$ in $f$. 
Due to the supporting property of $\Theta$,
as far as $\Theta$ is involved, 
we may always identify $f$ with $|x|$.

Recalling notation \eqref{251004}, we will often use a decomposition 
\begin{equation}
\alpha_jp_j=\alpha_fp_f+\alpha_j\ell_{jk}p_k
;\quad 
\ell_{jk}=\delta_{jk}-(\partial_jf)(\partial_kf)
.
\label{2507191540}
\end{equation}
In particular, we note that on $\mathop{\mathrm{supp}}\Theta$ 
the matrix $(\ell_{ij})_{i,j=1,\dots, d}$ represents the orthogonal projection onto the spherical direction,
and moreover 
\begin{equation}
\ell_{jk}=f(\partial_j\partial_kf)
\ \ \text{for }j,k=1,\dots,d
.
\label{eq:15091112}
\end{equation}

Now let us compute a commutator $\mathrm i[H,A_\Theta]=2\mathop{\mathrm{Im}}(A_\Theta H)$ 
for general $\Theta$. 

\begin{lemma}\label{25071822}
One has the identity 
\begin{align*}
\mathop{\mathrm{Im}}(A_\Theta H)
&=
\mathop{\mathrm{Re}}(\alpha_f\Theta'p_f)
+\mathop{\mathrm{Re}}(\alpha_jf^{-1}\Theta\ell_{jk}p_k\bigr)
\\&\phantom{{}={}}{}
-(\partial_f(q_0+q_1))\Theta
+(\Delta f)q_2\Theta 
+q_2\Theta' 
-2\mathop{\mathrm{Im}}(q_2\Theta p_f)
.
\end{align*}
\end{lemma}
\begin{remark}\label{251002}
In most of the later applications we will squeeze (weighted) positivity from the first and second terms 
on the right-hand side without energy cut-offs, 
in spite that they are first order differential operators
with different weights $\Theta'$ and $f^{-1}\Theta$. 
This is the most technical novelty of the paper.
The other terms along with the left-hand side are negligible.
For instance, 
the expectation of 
$\mathop{\mathrm{Im}}(A_\Theta H)=\mathop{\mathrm{Im}}(A_\Theta (H-\lambda))$ 
vanishes on an eigenstate $\phi$. 
\end{remark}
\begin{proof}
We decompose the left-hand side of the asserted identity as 
\begin{equation}
\mathop{\mathrm{Im}}(A_\Theta H)
=
\mathop{\mathrm{Im}}(A_\Theta \alpha_jp_j)
+\mathop{\mathrm{Im}}(A_\Theta (q_0+q_1))
+\mathop{\mathrm{Im}}(A_\Theta q_2)
.
\label{25071114}
\end{equation}
We can compute the first term on the right-hand side of \eqref{25071114}
by using \eqref{eq:15091112} as 
\begin{align*}
\mathop{\mathrm{Im}}(A_\Theta \alpha_jp_j)
&=
\mathop{\mathrm{Im}}\bigl((\Theta p_f+p_f^*\Theta)\alpha_jp_j\bigr)
\\&=
\mathop{\mathrm{Im}}(\alpha_j\Theta p_fp_j)
-\mathop{\mathrm{Im}}(\alpha_jp_j\Theta p_f)
\\&=
\mathop{\mathrm{Re}}(\alpha_f \Theta' p_f)
+\mathop{\mathrm{Re}}\bigl(\alpha_j\Theta(\partial_j\partial_kf) p_k\bigr)
\\&=
\mathop{\mathrm{Re}}(\alpha_f \Theta' p_f)
+\mathop{\mathrm{Re}}(\alpha_jf^{-1}\Theta\ell_{jk}p_k\bigr)
. 
\end{align*}
As for the second and third terms on the right-hand side of \eqref{25071114}, 
it is easy to see that 
\begin{equation*}
\mathop{\mathrm{Im}}(A_\Theta (q_0+q_1))
=-(\partial_f(q_0+q_1))\Theta,
\end{equation*}
and that 
\begin{align*}
\mathop{\mathrm{Im}}(A_\Theta q_2)
&
=\mathop{\mathrm{Im}}\bigl((\Theta p_f+p_f^*\Theta)q_2\bigr)
\\&
=
\mathop{\mathrm{Re}}\bigl((\partial_j\Theta\partial_jf) q_2\bigr)
+2\mathop{\mathrm{Im}}(p_f^*\Theta q_2)
\\&
=
(\Delta f) q_2\Theta
+q_2\Theta' 
-2\mathop{\mathrm{Im}}(q_2\Theta p_f)
.
\end{align*}
Hence we obtain the assertion. 
\end{proof}

We will also often use the following identity. 

\begin{lemma}\label{2507172236}
For any $z\in\mathbb C$ and $\Xi\in C^2(\mathbb R^d;\mathbb R)$ one can write 
\begin{align*}
p_j\Xi p_j
&
=
(z-q)^*(z-q)\Xi
-\mathop{\mathrm{Im}}\bigl(\alpha_j(\partial_j\Xi) (z-q)\bigr)
+\tfrac12(\Delta \Xi)
\\&\phantom{={}}{}
+\mathop{\mathrm{Re}}
\bigl(\bigl(2(z-q)^*\Xi+\mathrm i\alpha_j(\partial_j\Xi)\bigr)(H-z)\bigr)
+(H-z)^*\Xi(H-z)
.
\end{align*}
\end{lemma}
\begin{remark}
This implies that $p_j\Xi p_j$ is comparable to $(z-q)^*(z-q)\Xi\sim \Xi$.
In fact, the second and the third term are of lower order 
since $\Xi$ is differentiated, 
and the third and the fourth are negligible due to the factor $(H-z)$. 
\end{remark}
\begin{proof}
We can directly compute it as follows. 
In fact, by using \eqref{25071914} we have 
\begin{align*}
p_j\Xi p_j
&=
\mathop{\mathrm{Re}}(\Xi p_jp_j)
+\mathop{\mathrm{Im}}\bigl((\partial_j\Xi )p_j\bigr)
\\&
=
\mathop{\mathrm{Re}}(\alpha_j\alpha_k\Xi p_jp_k)
+\tfrac12(\Delta \Xi )
\\&
=p_j\alpha_j\Xi\alpha_kp_k
-\mathop{\mathrm{Im}}\bigl(\alpha_j(\partial_j\Xi)\alpha_kp_k\bigr)
+\tfrac12(\Delta \Xi )
\\&
=
(z-q+H-z)^*\Xi(z-q+H-z)
\\&\phantom{={}}{}
-\mathop{\mathrm{Im}}\bigl(\alpha_j(\partial_j\Xi) (z-q+H-z)\bigr)
+\tfrac12(\Delta \Xi)
\\&
=
(z-q)^*(z-q)\Xi
+2\mathop{\mathrm{Re}}\bigl((z-q)^*\Xi(H-z)\bigr)
+(H-z)^*\Xi(H-z)
\\&\phantom{={}}{}
-\mathop{\mathrm{Im}}\bigl(\alpha_j(\partial_j\Xi) (z-q)\bigr)
-\mathop{\mathrm{Im}}\bigl(\alpha_j(\partial_j\Xi)(H-z)\bigr)
+\tfrac12(\Delta \Xi)
.
\end{align*}
Thus obtain the assertion.
\end{proof}

\subsubsection{Weight functions for Rellich's theorem}

The weight function we use in this section is of the form 
\begin{align}
\Theta= \Theta_{a,b,\nu}^{\kappa,\mu}
=\chi_{a,b}\mathrm e^{\theta}
,
\label{eq:15.2.15.5.8bb}
\end{align}
which is dependent on parameters $a,b,\nu\in\mathbb N_0$ and $\kappa,\mu\ge 0$. 
Here $\chi_{a,b}\in C^\infty(\mathbb R^d)$ is a smooth cutoff function to a dyadic annular domain, defined as 
\begin{align}
\chi_{a,b}=\bar\chi_a\chi_b;\quad 
\bar\chi_a=1-\chi(f/2^a),\quad \chi_b=\chi(f/2^b),
\label{eq:11.7.11.5.14}
\end{align}  
with $\chi\in C^\infty(\mathbb R)$ being from \eqref{eq:14.1.7.23.24}. 
The exponent $\theta\in C^\infty(\mathbb R^d)$ is defined as 
\begin{equation}
\theta=\theta_\nu^{\kappa,\mu}
=2\kappa f+2\mu\int_0^f(1+\tau/2^{\nu})^{-1-\delta}\,\mathrm d\tau
.
\label{2507191448}
\end{equation}
Here and below the constant $\delta\in (0,\rho)$ is arbitrarily fixed, so that the dependence on it is suppressed
in \eqref{eq:15.2.15.5.8bb}. 
Note that the integral in \eqref{2507191448} is a refinement of the \textit{Yosida approximation}, satisfying 
\begin{equation*}
0< 
\int_0^f(1+\tau/2^{\nu})^{-1-\delta}\,\mathrm d\tau
\le \min\{f,\delta^{-1} 2^{\nu}\}\ \ \text{for any }\nu\in\mathbb N_0,
\end{equation*}
and 
\begin{equation*}
\int_0^f(1+\tau/2^{\nu})^{-1-\delta}\,\mathrm d\tau
\uparrow f\ \ \text{pointwise as }\nu\to\infty
.
\end{equation*}
We note that, if we denote the derivatives in $f$ by primes as before, then 
\begin{align*}
\begin{split}
\theta'=2\kappa+2\mu(1+f/2^{\nu})^{-1-\delta},\quad
\theta''=-2(1+\delta)\mu 2^{-\nu}(1+f/2^{\nu})^{-2-\delta},\quad \ldots. 
\end{split}
\end{align*} 
In particular, noting $2^{-\nu}(1+f/2^{\nu})^{-1}\leq f^{-1}$, we can bound 
for any $k=2,3,\ldots$
\begin{align*}
|(\theta-2\kappa f)'|\leq C_{\delta,1}(\theta'-2\kappa)
,\qquad 
|(\theta-2\kappa f)^{(k)}|\leq C_{\delta,k}f^{2-k}|\theta''|
.
\end{align*}

\subsection{A priori super-exponential decay estimate}\label{subsec:150909511}

Now we present a key lemma for the proof of Proposition~\ref{prop:absence-eigenvalues-1b}. 

\begin{lemma}\label{25071915}
Let $\lambda\gtrless m_\pm$, respectively, and fix any $\kappa_0\ge 0$. 
Then there exist $\mu,c,C>0$ and $a_0\in\mathbb N_0$ such that uniformly in $\kappa\in [0,\kappa_0]$, $b>a\ge a_0$ and $\nu\in\mathbb N_0$
\begin{align*}
\pm\mathop{\mathrm{Im}}\bigl(A_\Theta(H-\lambda)\bigr)
&\ge 
c\min\{f^{-1},\theta'\}\Theta
-C\bigl(\chi_{a-1,a+1}^2+\chi_{b-1,b+1}^2\big)f^{-1}\mathrm e^\theta
\\&\phantom{={}}{}
-(H-\lambda)\gamma(H-\lambda),
\end{align*}
respectively, 
where $\gamma=\gamma_{a,b}$ is a certain function, 
independent of $\kappa$ and $\nu$, satisfying 
$\mathop{\mathrm{supp}}\gamma\subset\mathop{\mathrm{supp}}\chi_{a,b}$
and 
$|\gamma|\le C_{a,b}$.
\end{lemma}
\begin{remarks}\label{26020817}
\begin{enumerate}
\item
The first term on the right-hand side is from the first and second term of 
Lemma~\ref{25071822} with $\Theta$ of the form \eqref{eq:15.2.15.5.8bb}. 
It is exactly the minimum of the coefficients
\item 
Let us outline the deduction of Proposition~\ref{prop:absence-eigenvalues-1b} from Lemma~\ref{25071915}: 
Assume 
\begin{equation}
\kappa_0=\sup\{\kappa\ge 0;\ \mathrm e^{\kappa f}\phi\in \mathcal B^*_0\}<\infty,  
\label{2602081719}
\end{equation}
let $\mu>0$ be as above, 
and let $\kappa=0$ if $\kappa_0=0$, and $\kappa\in [0,\kappa_0)$ with $\kappa+\mu>\kappa_0$ otherwise; 
Take the expectation of the above inequality on the eigenstate $\phi$,
and the left-hand side and the last term on the right-hand side vanish; 
Let $b\to\infty$, and a contribution from 
$\chi_{b-1,b+1}^2f^{-1}\mathrm e^\theta$ vanishes by definition of 
$\mathcal B^*_0$; Let $\nu\to\infty$,
which replaces the Yosida approximation $\theta$ by $f$; 
Then we obtain $\mathrm e^{(\kappa+\mu) f}\phi\in \mathcal B^*_0$, which contradicts the assumption. 
Proposition~\ref{prop:absence-eigenvalues-1b} thus follows. 
The second term on the right-hand side of Lemma~\ref{25071822} may be seen as a
boundary contribution, as in Gauss's divergence theorem. 
\end{enumerate}
\end{remarks}

\begin{proof}
Let us only discuss the upper sign since the lower one can be done by the same manner. 
Fix $\lambda> m_+$ and $\kappa_0\ge 0$ as in the assertion.
For the moment we merely choose $a_0\in \mathbb N_0$ such that for some $c_1>0$
\begin{equation}
\min\{\lambda-q_0,\lambda-q_0-q_1\}\ge c_1,
\ \ \text{uniformly in }\{x\in\mathbb R^d;\ |x|\ge 2^{a_0}\}. 
\label{2507182241b}
\end{equation}
Then, except for the last part of the proof,
the below estimates are all uniform in $\kappa\in [0,\kappa_0]$, $\mu\in (0,1]$, $b>a\ge a_0$ and $\nu\in\mathbb N_0$,
and thus until then $c_*,C_*>0$ are independent of them.  
Only in the very last step we shall restrict ranges of all these parameters to verify the assertion,
where we may also retake $a_0\in\mathbb N_0$ larger if necessary.

In this proof we gather and absorb \textit{admissible error terms} into   
\begin{align}
\begin{split}
Q&=
\bigl(\chi_{a,b}|\theta''|+\chi_{a,b}f^{-1-\rho}+|\chi_{a,b}'|\bigr)\mathrm e^\theta
\\&\phantom{{}={}}{}
+p_j\bigl(\chi_{a,b}|\theta''|+\chi_{a,b}f^{-1-\rho}+|\chi_{a,b}'|\bigr)\mathrm e^\theta p_j
+(H-\lambda)f^{1+\rho}\Theta(H-\lambda)
.
\end{split}\label{eq:15091119b}
\end{align}
We will later prove that it is in fact negligible. 

Now let us start to compute the left-hand side of the asserted inequality.
By Lemma~\ref{25071822} and the Cauchy--Schwarz inequality we first have 
\begin{align}
\begin{split}
\mathop{\mathrm{Im}}\bigl(A_\Theta(H-\lambda)\bigr)
&=
\mathop{\mathrm{Re}}\bigl(\alpha_f\theta'\Theta p_f\bigr)
+ \mathop{\mathrm{Re}}\bigl(\alpha_f\chi_{a,b}'\mathrm e^\theta p_f\bigr)
+ \mathop{\mathrm{Re}}\bigl(\alpha_jf^{-1}\Theta\ell_{jk}p_k\bigr)
\\&\phantom{{}={}}{}
- \bigl(\partial_f(q_0+q_1)\bigr)\Theta
+ (\Delta f)q_2\Theta 
+ q_2\Theta' 
- 2\mathop{\mathrm{Im}}(q_2\Theta p_f)
\\&
\ge 
\mathop{\mathrm{Re}}\bigl(\alpha_f\theta'\Theta p_f\bigr)
+ \mathop{\mathrm{Re}}\bigl(\alpha_jf^{-1}\Theta\ell_{jk}p_k\bigr)
-C_1Q
.
\end{split}
\label{25071722b}
\end{align}
Recalling the expression \eqref{2507191540} of $\ell_{jk}$, 
and using the Cauchy--Schwarz inequality and \eqref{2507182241b}, 
we can combine and bound the first and second terms on the right-hand side of \eqref{25071722b} as 
\begin{align}
\begin{split}
&
\mathop{\mathrm{Re}}\bigl(\alpha_f\theta'\Theta p_f\bigr)
+ \mathop{\mathrm{Re}}\bigl(\alpha_jf^{-1}\Theta\ell_{jk}p_k\bigr)
\\&
=
\tfrac12\mathop{\mathrm{Re}}\bigl(\alpha_f(\theta'+f^{-1})\Theta p_f\bigr)
-\tfrac12\mathop{\mathrm{Re}}\bigl(\alpha_f(f^{-1}-\theta')\Theta p_f\bigr)
\\&\phantom{{}={}}
+\tfrac12\mathop{\mathrm{Re}}\bigl(\alpha_j(\theta'+f^{-1})\Theta\ell_{jk}p_k\bigr)
+\tfrac12\mathop{\mathrm{Re}}\bigl(\alpha_j(f^{-1}-\theta')\Theta\ell_{jk}p_k\bigr)
\\&
=
\tfrac12\mathop{\mathrm{Re}}\bigl(\alpha_j(\theta'+f^{-1})\Theta p_j\bigr)
-\tfrac12\mathop{\mathrm{Re}}\bigl(\alpha_j(f^{-1}-\theta')\Theta h_{jk}p_k\bigr)
\\&
=
\tfrac12(\lambda-q)(\theta'+f^{-1})\Theta 
+\tfrac12\mathop{\mathrm{Re}}\bigl((\theta'+f^{-1})\Theta (H-\lambda)\bigr)
\\&\phantom{{}={}}
-\tfrac12\mathop{\mathrm{Re}}\bigl(\chi_{a,b}(f^{-1}-\theta')f^{2-d}\alpha_jh_{jk}p_kf^{d-2}\mathrm e^{\theta}\bigr)
\\&
\ge 
\tfrac12(\lambda-q_0-q_1)(\theta'+f^{-1})\Theta 
- \tfrac14(\lambda-q_0)|f^{-1}-\theta'|\Theta 
\\&\phantom{{}={}}
- \tfrac14\mathrm e^{\theta}f^{d-2}p_ih_{il} \alpha_l\chi_{a,b}(\lambda-q_0)^{-1}|f^{-1}-\theta'|f^{4-2d}\mathrm e^{-\theta}\alpha_jh_{jk}p_kf^{d-2}\mathrm e^{\theta}
\\&\phantom{{}={}}
-C_2Q
,
\end{split}
\label{2507182240b}
\end{align}
where we have set 
\begin{equation}
h_{jk}=(\partial_jf)(\partial_kf)-\ell_{jk}.
\label{25072714}
\end{equation}
We are going to further compute and bound the third term on the right-hand side of \eqref{2507182240b},
which requires a subtle treatment. 
Let us set for short
\begin{align*}
\Xi&=\chi_{a,b}(\lambda-q_0)^{-1}\bigl((f^{-1}-\theta')^2+f^{-2-2\rho}\bigr)^{1/2},\\
\widetilde\Xi&=\chi_{a,b}(\lambda-\widetilde q_0)^{-1}\bigl((f^{-1}-\theta')^2+f^{-2-2\rho}\bigr)^{1/2}
,
\end{align*}
cf.\ \eqref{250719},  and we proceed with the third term of \eqref{2507182240b} 
by using \eqref{25072714}, \eqref{2507191540} and $\alpha_f^2=1$ 
as 
\begin{align}
\begin{split}
&
- \tfrac14\mathrm e^{\theta}f^{d-2}p_ih_{il} \alpha_l\chi_{a,b}(\lambda-q_0)^{-1}|f^{-1}-\theta'|f^{4-2d}\mathrm e^{-\theta}\alpha_jh_{jk}p_kf^{d-2}\mathrm e^{\theta}
\\&
\ge 
- \tfrac14\mathrm e^{\theta}f^{d-2}p_ih_{il} \alpha_l\Xi f^{4-2d}\mathrm e^{-\theta}\alpha_jh_{jk}p_kf^{d-2}\mathrm e^{\theta}
\\&
=
- \tfrac14\mathrm e^{\theta/2}\bigl(p_i+\tfrac{\mathrm i}2(\partial_if)\theta'+\mathrm i(d-2)(\partial_if)f^{-1}\bigr)h_{il} \alpha_l\Xi\alpha_jh_{jk}
\\&\phantom{{}={}}{}
\cdot 
\bigl(p_k-\tfrac{\mathrm i}2(\partial_kf)\theta'-\mathrm i(d-2)(\partial_kf)f^{-1}\bigr)\mathrm e^{\theta/2}
\\&
=
- \tfrac14\mathrm e^{\theta/2}p_ih_{il} \alpha_l\Xi\alpha_jh_{jk}p_k\mathrm e^{\theta/2}
\\&\phantom{{}={}}{}
+\tfrac14\mathop{\mathrm{Im}}\bigl(\alpha_f\Xi(\theta'+2(d-2) f^{-1})\mathrm e^{\theta/2}\alpha_jh_{jk}p_k\mathrm e^{\theta/2}\bigr)
\\&\phantom{{}={}}{}
- \tfrac1{16}\widetilde\Xi(\theta'+2(d-2)f^{-1})^2\mathrm e^{\theta}
\\&
\ge 
- \tfrac14\mathop{\mathrm{Re}}\bigl(\alpha_l\alpha_j\mathrm e^{\theta/2}\widetilde\Xi h_{il} h_{jk}p_ip_k\mathrm e^{\theta/2}\bigr)
- \tfrac14\mathop{\mathrm{Im}}\bigl(\alpha_l\alpha_j\mathrm e^{\theta/2}(\partial_i\widetilde\Xi h_{il} h_{jk})p_k\mathrm e^{\theta/2}\bigr)
\\&\phantom{{}={}}{}
+\tfrac14\mathop{\mathrm{Im}}\bigl(\widetilde\Xi (\theta'+2(d-2) f^{-1})\mathrm e^{\theta/2}p_f\mathrm e^{\theta/2}\bigr)
\\&\phantom{{}={}}{}
-\tfrac14\mathop{\mathrm{Im}}\bigl(\alpha_f\Xi(\theta'+2(d-2) f^{-1})\mathrm e^{\theta/2}\alpha_j\ell_{jk}p_k\mathrm e^{\theta/2}\bigr)
\\&\phantom{{}={}}{}
- \tfrac1{16}\widetilde\Xi\theta'^2\mathrm e^{\theta}
-\tfrac{d-2}4\widetilde\Xi f^{-1}\theta' \mathrm e^{\theta}
-C_3Q
.
\end{split}
\label{250726}
\end{align}
We continue to compute \eqref{250726}. 
Note that by \eqref{25072714}, \eqref{2507191540} and \eqref{eq:15091112}
\begin{equation}
\partial_ih_{jk}=2f^{-1}\ell_{ij}(\partial_kf)+2f^{-1}(\partial_jf)\ell_{ik}
,\quad 
(\Delta f)=(d-1)f^{-1}
.
\label{260208}
\end{equation}
Then, after some computations employing 
\eqref{25071914}, Assumption~\ref{2571017}, the Cauchy--Schwarz inequality, \eqref{260208} and 
$\ell_{ij}\alpha_i\alpha_j=(d-1)$ along with \eqref{2507191540}, \eqref{eq:15091112}, \eqref{25072714} and $\alpha_f^2=1$ 
as above, we obtain 
\begin{align}
\begin{split}
&
- \tfrac14\mathrm e^{\theta}f^{d-2}p_ih_{il} \alpha_l\chi_{a,b}(\lambda-q_0)^{-1}|f^{-1}-\theta'|f^{4-2d}\mathrm e^{-\theta}\alpha_jh_{jk}p_kf^{d-2}\mathrm e^{\theta}
\\&
\ge 
- \tfrac14\mathop{\mathrm{Re}}\bigl(\mathrm e^{\theta/2}\widetilde\Xi h_{ij} h_{jk}p_ip_k\mathrm e^{\theta/2}\bigr)
- \tfrac{d-1}2\mathop{\mathrm{Im}}\bigl(\alpha_f\alpha_j\widetilde\Xi f^{-1}\mathrm e^{\theta/2} h_{jk}p_k\mathrm e^{\theta/2}\bigr)
\\&\phantom{{}={}}{}
- \tfrac12\mathop{\mathrm{Im}}\bigl(\alpha_l\alpha_j\widetilde\Xi f^{-1}\mathrm e^{\theta/2}h_{il}\ell_{ij}p_f\mathrm e^{\theta/2}\bigr)
- \tfrac12\mathop{\mathrm{Im}}\bigl(\alpha_l\alpha_f\widetilde\Xi f^{-1}\mathrm e^{\theta/2}h_{il}\ell_{ik}p_k\mathrm e^{\theta/2}\bigr)
\\&\phantom{{}={}}{}
+\tfrac{d-1}8\widetilde\Xi f^{-1}\theta'\mathrm e^{\theta}
-\tfrac14\mathop{\mathrm{Im}}\bigl(\alpha_f\Xi(\theta'+2(d-2) f^{-1})\mathrm e^{\theta}\alpha_j\ell_{jk}p_k\bigr)
\\&\phantom{{}={}}{}
- \tfrac1{16}\widetilde\Xi\theta'^2\mathrm e^{\theta}
-\tfrac{d-2}4\widetilde\Xi f^{-1}\theta' \mathrm e^{\theta}
-C_4Q
\\&
\ge 
- \tfrac14\mathop{\mathrm{Re}}\bigl(\mathrm e^{\theta/2}\widetilde\Xi \delta_{ik}p_ip_k\mathrm e^{\theta/2}\bigr)
-\tfrac14\mathop{\mathrm{Im}}\bigl(\alpha_f\Xi\theta'\mathrm e^{\theta}\alpha_j\ell_{jk}p_k\bigr)
\\&\phantom{{}={}}{}
- \tfrac1{16}\widetilde\Xi\theta'^2\mathrm e^{\theta}
-\tfrac{d-3}8\widetilde\Xi f^{-1}\theta' \mathrm e^{\theta}
-C_5Q
.
\end{split}
\label{250726f}
\end{align}
Let us compute the first term on the right-hand side of \eqref{250726f} as follows: 
\begin{align}
\begin{split}
&
- \tfrac14\mathop{\mathrm{Re}}\bigl(\widetilde\Xi \mathrm e^{\theta/2}\delta_{ik}p_ip_k\mathrm e^{\theta/2}\bigr)
\\&
=- \tfrac14\mathop{\mathrm{Re}}
\bigl(\widetilde\Xi \mathrm e^{\theta/2}\alpha_i\alpha_kp_ip_k\mathrm e^{\theta/2}\bigr)
\\&
=
- \tfrac14\mathrm e^{\theta/2}p_i\alpha_i\Xi \alpha_kp_k\mathrm e^{\theta/2}
+ \tfrac14\mathop{\mathrm{Im}}
\bigl(\mathrm e^{\theta/2}\alpha_i (\partial_i\Xi) \alpha_kp_k\mathrm e^{\theta/2}\bigr)
\\&
\ge 
- \tfrac14\bigl(p_i\alpha_i+\tfrac{\mathrm i}2\alpha_f\theta'\bigr)\Xi \mathrm e^{\theta}
\bigl(\alpha_kp_k-\tfrac{\mathrm i}2\alpha_f\theta'\bigr)
-C_6Q
\\&
=
- \tfrac14p_i\alpha_i\Xi \mathrm e^{\theta}\alpha_kp_k
+\tfrac14\mathop{\mathrm{Im}}\bigl(\alpha_f\Xi \theta'\mathrm e^{\theta}\alpha_kp_k\bigr)
-\tfrac1{16}\widetilde\Xi\theta'^2 \mathrm e^{\theta}
-C_6Q
\\&
=
- \tfrac14(\lambda-q+H-\lambda)\Xi \mathrm e^{\theta}(\lambda-q+H-\lambda)
+\tfrac14\mathop{\mathrm{Im}}\bigl(\alpha_f\Xi \theta'\mathrm e^{\theta}\alpha_fp_f\bigr)
\\&\phantom{{}={}}
+\tfrac14\mathop{\mathrm{Im}}\bigl(\alpha_f\Xi \theta'\mathrm e^{\theta}\alpha_j\ell_{jk}p_k\bigr)
-\tfrac{1}{16}\widetilde\Xi \theta'^2\mathrm e^{\theta}
-C_6Q
\\&
\ge 
- \tfrac14(\lambda-q)\Xi \mathrm e^{\theta}(\lambda-q)
+\tfrac18\bigl(\partial_i(\partial_if)\widetilde\Xi \theta'\mathrm e^{\theta}\bigr)
+\tfrac14\mathop{\mathrm{Im}}\bigl(\alpha_f\Xi \theta'\mathrm e^{\theta}\alpha_j\ell_{jk}p_k\bigr)
\\&\phantom{{}={}}
-\tfrac{1}{16}\widetilde\Xi \theta'^2\mathrm e^{\theta}
-C_7Q
\\&
\ge 
- \tfrac14(\lambda-q_0-2q_1)|f^{-1}-\theta'|\Theta
+\tfrac14\mathop{\mathrm{Im}}\bigl(\alpha_f\Xi \theta'\mathrm e^{\theta}\alpha_j\ell_{jk}p_k\bigr)
\\&\phantom{{}={}}
+\tfrac{1}{16}\widetilde\Xi \theta'^2\mathrm e^{\theta}
+\tfrac{d-1}8\widetilde\Xi f^{-1}\theta'\mathrm e^{\theta}
-C_8Q
.
\end{split}
\label{25072622}
\end{align}
Hence by \eqref{25071722b}, \eqref{2507182240b}, \eqref{250726f}, \eqref{25072622}, and \eqref{2507182241b} we obtain  
\begin{align}
\begin{split}
\mathop{\mathrm{Im}}\bigl(A_\Theta(H-\lambda)\bigr)
&
\ge 
c_2\bigl(\min\{f^{-1},\theta'\}+|f^{-1}-\theta'|f^{-1}\theta'\bigr)\Theta
-C_9Q
\\&\ge 
c_2\min\{f^{-1},\theta'\}\Theta
-C_9Q
.
\end{split}
\label{2507182240bb}
\end{align}
Note that we have dropped a positive contribution since it is not necessary here. 
However, it is essential in the proof of Proposition~\ref{prop:absence-eigenvalues-1bbb}. 

Next we discuss $Q$. We claim that 
\begin{align}
\begin{split}
Q&\le 
C_{10}\bigl(|\theta''|+f^{-1-\rho}\bigr)\Theta
+C_{10}\bigl(\chi_{a-1,a+1}^2+\chi_{b-1,b+1}^2\big)f^{-1}\mathrm e^\theta
\\&\phantom{{}={}}{}
+C_{10}(H-\lambda)f^{1+\rho}\Theta(H-\lambda)
.
\end{split}
\label{250727}
\end{align}
For that it suffices to consider the second term of \eqref{eq:15091119b}
since the contribution from $|\chi_{a,b}'|\mathrm e^\theta$ is obviously absorbed into the 
second term of \eqref{250727}. 
However, the second term of \eqref{eq:15091119b} is obviously absorbed into the 
right-hand side of \eqref{250727} due to Lemma~\ref{2507172236} 
and the Cauchy--Schwarz inequality.
Thus we obtain \eqref{250727}. 

Therefore by \eqref{2507182240bb} and \eqref{250727} it follows that 
\begin{align}
\begin{split}
\mathop{\mathrm{Im}}\bigl(A_\Theta(H-\lambda)\bigr)
&
\ge 
c_2\bigl(\min\{f^{-1},\theta'\}-C_{11}|\theta''|-C_{11}f^{-1-\rho}\bigr)\Theta
\\&\phantom{{}={}}{}
-C_{11}\bigl(\chi_{a-1,a+1}^2+\chi_{b-1,b+1}^2\big)f^{-1}\mathrm e^\theta
\\&\phantom{{}={}}{}
-C_{11}(H-\lambda)f^{1+\rho}\Theta(H-\lambda)
,
\end{split}
\label{2507182240bbf}
\end{align}
and finally it suffices to look at the coefficient of the first term on the right-hand side of \eqref{2507182240bbf}. 
However, it is straightforward by restricting the ranges of parameters. 
We first choose sufficiently small $\mu\in (0,1]$, and then retake $a_0\in\mathbb N_0$ larger if necessary.
Then we can verify that on $\mathop{\mathrm{supp}}\Theta=\mathop{\mathrm{supp}}\chi_{a,b}$ 
\begin{align*}
\min\{f^{-1},\theta'\}-C_{11}|\theta''|-C_{11}f^{-1-\rho}
\ge c_3\min\{f^{-1},\theta'\}.
\end{align*}
Thus we are done with the proof. 
\end{proof}

\begin{proof}[Proof of Proposition~\ref{prop:absence-eigenvalues-1b}]
Let $\phi\in\mathcal B^*_0$ and $\lambda\in\mathbb R\setminus[m_-,m_+]$ be as in the assertion, 
and we follow the strategy outlined in Remarks~\ref{26020817}. 
Let $\kappa_0\in[0,\infty]$ be from \eqref{2602081719}, 
let $\mu>0$ and $a_0\in\mathbb N_0$ be from Lemma~\ref{25071915},
and choose $\kappa\ge 0$ as in Remarks~\ref{26020817}.
Then, taking the expectation of the inequality from Lemma~\ref{25071915} 
on the state $\chi_{a-2,b+2}\phi$, 
we obtain that for $b>a\ge a_0$ and $\nu\in\mathbb N_0$
\begin{align}
\begin{split}
\bigl\|(\min\{f^{-1},\theta'\}\Theta)^{1/2}\phi\bigr\|^2
&
\le 
C_a\|\chi_{a-1,a+1}\phi\|^2
+C_\nu 2^{-b}\|\chi_{b-1,b+1}\mathrm e^{\kappa f}\phi\|^2
\end{split}
\label{eq:11.7.16.3.22a}
\end{align}
with $C_a$ being independent of $b,\nu$, and $C_\nu$ of $a,b$. 
Now we take the limit $b\to\infty$ in \eqref{eq:11.7.16.3.22a}. 
Since the second term on the right-hand side vanishes 
due to the assumption $\mathrm e^{\kappa f}\phi\in \mathcal B^*_0$,
it follows by Lebesgue's monotone convergence theorem that 
\begin{align}
\bigl\|(\bar\chi_a \min \{f^{-1},\theta'\}\mathrm{e}^{\theta})^{1/2}\phi\bigr\|^2
 &\le 
C_a\|\chi_{a-1,a+1}\phi\|^2
\label{eq:11.7.16.3.43a}
\end{align}
with $\bar\chi_{a}$ being from \eqref{eq:11.7.11.5.14}. 
Next we let $\nu \to\infty$ in \eqref{eq:11.7.16.3.43a},
and then by Lebesgue's monotone convergence theorem again it follows that 
\[
\bar\chi_a^{1/2}f^{-1/2}\mathrm{e}^{(\kappa+\mu) f}\phi\in \mathcal H.
\] 
This implies $\mathrm e^{(\kappa+\mu) f}\phi\in L^2_{-1/2}\subset \mathcal B^*_0$,
and thus a contradiction. We are done.
\end{proof}

\subsection{Absence of super-exponentially decaying eigenstates}\label{subsec:150909513}

To prove Proposition~\ref{prop:absence-eigenvalues-1bbb} we present to the following key lemma. 

\begin{lemma}\label{25071915b}
Let $\lambda\gtrless m_\pm$, respectively, and fix $\mu=0$
in the definition \eqref{eq:15.2.15.5.8bb} of $\Theta$, 
so that $\theta=2\kappa f$, and that $\Theta$ is independent of $\nu\in\mathbb N_0$.
Then there exist $c,C>0$ and $a_0\in\mathbb N_0$ such that uniformly in $\kappa\ge 1$
and $b>a\ge a_0$
\begin{align*}
\pm\mathop{\mathrm{Im}}\bigl(A_\Theta(H-\lambda)\bigr)
&\ge 
c\kappa^2f^{-1}\Theta
-C\kappa^2\bigl(\chi_{a-1,a+1}^2+\chi_{b-1,b+1}^2\big)f^{-1}\mathrm e^{2\kappa f}
\\&\phantom{={}}{}
-(H-\lambda)\gamma(H-\lambda),
\end{align*}
respectively, 
where $\gamma=\gamma_{a,b,\kappa}$ is a certain function
satisfying $\mathop{\mathrm{supp}}\gamma\subset\mathop{\mathrm{supp}}\chi_{a,b}$ 
and $|\gamma|\le C_{a,b,\kappa}$.
\end{lemma}

\begin{proof}
The proof is very similar to that of Lemma~\ref{25071915}, 
and we present only key steps, omitting details of the computations. 
Note that here we have to take particular care of dependence on the parameter $\kappa\ge 1$, 
and we squeeze positivity from a different term from Lemma~\ref{25071915}. 
We only discuss the upper sign. 
Fix any $\lambda> m_+$ and $\mu=0$ as in the assertion.
We first choose $a_0\in \mathbb N_0$ such that for some $c_1>0$
\begin{equation}
\min\{\lambda-q_0,\lambda-q_0-q_1\}\ge c_1\ \ \text{and}\ \ 
\theta'-f^{-1}= 2\kappa-f^{-1}\ge c_1\kappa
\label{2507182241bb}
\end{equation}
uniformly in $\{x\in\mathbb R^d;\ |x|\ge 2^{a_0}\}$ and $\kappa\ge 1$,
but later we may retake it larger. 
All below estimates are uniform in $\kappa\ge 1$ and $b>a\ge a_0$.

Now let us start to compute the left-hand side of the asserted inequality. 
In this proof we set    
\begin{align}
\begin{split}
Q&=
\kappa^2\bigl(\chi_{a,b}f^{-1-\rho}+|\chi_{a,b}'|\bigr)\mathrm e^{2\kappa f}
+p_j\bigl(\chi_{a,b}f^{-1-\rho}+|\chi_{a,b}'|\bigr)\mathrm e^{2\kappa f} p_j
\\&\phantom{{}={}}
+\kappa(H-\lambda)f^{1+\rho}\Theta(H-\lambda)
.
\end{split}\label{beq:15091119b}
\end{align}
Then, similarly to \eqref{25071722b} and \eqref{2507182240b}, we can deduce, noting also \eqref{2507182241bb},  
\begin{align}
\begin{split}
\mathop{\mathrm{Im}}\bigl(A_\Theta(H-\lambda)\bigr)
&
\ge 
\tfrac12(\lambda-q_0-q_1)(2\kappa+f^{-1})\Theta 
- \tfrac14(\lambda-q_0)(2\kappa-f^{-1})\Theta 
\\&\phantom{{}={}}
- \tfrac14\mathrm e^{2\kappa f}f^{d-2}p_ih_{il} \alpha_l\Xi f^{4-2d}\mathrm e^{-2\kappa f}\alpha_jh_{jk}p_kf^{d-2}\mathrm e^{2\kappa f}
\\&\phantom{{}={}}
-C_1Q
,
\end{split}
\label{2507182240bbc}
\end{align}
where $h_{jk}$ is defined as \eqref{25072714}, and 
\begin{equation*}
\Xi=\chi_{a,b}(\lambda-q_0)^{-1}(2\kappa-f^{-1}),\quad 
\widetilde\Xi=\chi_{a,b}(\lambda-\widetilde q_0)^{-1}(2\kappa-f^{-1})
.
\end{equation*}
As for the third term on the right-hand side of \eqref{2507182240bbc}, 
we repeat the arguments as in \eqref{250726}, \eqref{250726f} and \eqref{25072622}, 
and then we obtain 
\begin{align}
\begin{split}
&
- \tfrac14\mathrm e^{2\kappa f}f^{d-2}p_ih_{il} \alpha_l\widetilde\Xi f^{4-2d}\mathrm e^{-2\kappa f}\alpha_jh_{jk}p_kf^{d-2}\mathrm e^{2\kappa f}
\\&
\ge 
- \tfrac14(\lambda-q_0-2q_1)(2\kappa-f^{-1})\Theta
+\tfrac{\kappa}2\widetilde\Xi f^{-1}\mathrm e^{2\kappa f}
-C_3Q
.
\end{split}
\label{250726b}
\end{align}
Hence by \eqref{2507182240bbc} and \eqref{250726b} it follows that 
\begin{align}
\begin{split}
\mathop{\mathrm{Im}}\bigl(A_\Theta(H-\lambda)\bigr)
&
\ge 
c_2\kappa^2 f^{-1}\Theta
-C_4Q
.
\end{split}
\label{2507182240bbb}
\end{align}
Finally, similarly to \eqref{250727}, we can control $Q$ by Lemma~\ref{2507172236} as 
\begin{align}
\begin{split}
Q&\le 
C_5\kappa^2f^{-1-\rho}\Theta
+C_5\kappa^2\bigl(\chi_{a-1,a+1}^2+\chi_{b-1,b+1}^2\big)f^{-1}\mathrm e^{2\kappa f}
\\&\phantom{{}={}}{}
+C_5\kappa (H-\lambda)f^{1+\rho}\Theta(H-\lambda)
,
\end{split}
\label{250727b}
\end{align}
so that by \eqref{2507182240bbb} and \eqref{250727b}
\begin{align*}
\begin{split}
\mathop{\mathrm{Im}}\bigl(A_\Theta(H-\lambda)\bigr)
&
\ge 
(c_3-C_6f^{-\rho})\kappa^2f^{-1}\Theta
+C_6\kappa^2\bigl(\chi_{a-1,a+1}^2+\chi_{b-1,b+1}^2\big)f^{-1}\mathrm e^{2\kappa f}
\\&\phantom{{}={}}{}
+C_6\kappa (H-\lambda)f^{1+\rho}\Theta(H-\lambda)
.
\end{split}
\end{align*}
Therefore, letting $a_0\in\mathbb N_0$ larger if necessary, we are done with the proof. 
\end{proof}

\begin{proof}[Proof of Proposition~\ref{prop:absence-eigenvalues-1bbb}]
Let $\phi\in\mathcal B^*_0$ and $\lambda\in\mathbb R\setminus[m_-,m_+]$ be as in the assertion. 
Let $a_0\in\mathbb N_0$ be from Lemma~\ref{25071915b},
and taking the expectation of the inequality from Lemma~\ref{25071915b} on the state $\chi_{a-2,b+2}\phi$. 
Then we have for any $\kappa\ge 1$ and $b>a\ge a_0$
\begin{equation}
\|\chi_{a,b}^{1/2}\mathrm e^{\kappa f}\phi\|^2
\leq 
C\|\chi_{a-1,a+1}\mathrm e^{\kappa f}\phi\|^2
+C 2^{-b}\|\chi_{b-1,b+1}\mathrm e^{\kappa f}\phi\|^2.
\label{eq:11.7.16.3.22ab}
\end{equation}
If we let $b\to\infty$ in \eqref{eq:11.7.16.3.22ab},  
the second term on the right-hand side vanishes 
since $\mathrm e^{\kappa f}\phi\in \mathcal B^*_0$, 
and thus by Lebesgue's monotone convergence theorem 
\[
\bigl\|\bar\chi_a^{1/2} \mathrm{e}^{\kappa f}\phi\bigr\|^2
\le 
C\|\mathrm{e}^{\kappa f}\chi_{a-1,a+1}\phi\|^2,
\]
or 
\begin{equation}
\bigl\|\bar\chi_a^{1/2} \mathrm{e}^{\kappa (f-2^{a+2})}\phi\bigr\|^2
\le 
C\|\chi_{a-1,a+1}\phi\|^2,
\label{eq:11.7.16.3.43ab}
\end{equation}
with $\bar\chi_{a}$ being from \eqref{eq:11.7.11.5.14}. 
Now we assume $\bar\chi_{a+2}\phi\not\equiv 0$ on $\mathbb R^d$.
Then the left-hand side of \eqref{eq:11.7.16.3.43ab} grows exponentially as $\kappa\to\infty$, 
whereas the right-hand side remains bounded.
This is a contradiction, and thus $\bar\chi_{b+2}\phi\equiv 0$. 
Then by the unique continuation property \cite{MR1741374} we conclude that $\phi\equiv 0$ on $\mathbb R^d$. 
We are done. 
\end{proof}

\section{Proof of LAP bounds}

In this section prove Theorem~\ref{thm:12.7.2.7.9}. 
The proof again relies on commutator arguments \cite{MR4062329} with conjugate operator 
$A_\Theta$ from \eqref{260728}, however, 
with a different weight function $\Theta$ from that for Rellich's theorem.
We in fact use
\begin{align}
\Theta=\Theta_{a,\nu}=\bar\chi_a\theta,\quad 
\theta=\theta_\nu=\int_0^{f/2^\nu}(1+\tau)^{-1-\delta}\,\mathrm d\tau
=\delta^{-1}\bigl(1-(1+f/2^\nu)^{-\delta}\bigr), 
\label{eq:15.2.15.5.8}
\end{align}
depending on parameters $a,\nu\in\mathbb N_0$, see \eqref{eq:11.7.11.5.14}
for the definition of $\bar\chi_a$.
Here $\delta\in (0,\rho)$ is an arbitrarily fixed constant,
and the dependence on it is suppressed 
since we shall not retake it in this section.

In Section~\ref{subsec:Improved radiation conditionsb}
we present preliminary commutator computations.  
In Section~\ref{subsec:15.2.14.14.41} 
we prove Theorem~\ref{thm:12.7.2.7.9} by contradiction 
to Rellich's theorem. 
Throughout the section we assume Assumption~\ref{2571017}.

\subsection{Preliminaries}
\label{subsec:Improved radiation conditionsb}

\subsubsection{Properties of weight function}
Here we present a key commutator estimate for the proof of Theorem~\ref{thm:12.7.2.7.9},
but before that  
let us recall some elementary properties of the function $\theta$ from \eqref{eq:15.2.15.5.8}.
As in the previous section we denote the derivatives in $f$ by primes,
such as 
\begin{align}
\theta'=2^{-\nu}(1+f/2^\nu)^{-1-\delta},\quad
\theta''=-2^{-2\nu}(1+\delta)(1+f/2^\nu{})^{-2-\delta}.
\label{eq:15.2.15.5.9}
\end{align}

\begin{lemma}\label{lem:15.2.18.14.28}
There exist $c,C>0$ such that uniformly in $\nu\in\mathbb N_0$  
\begin{align*}
&c/2^\nu{} \le \theta\le \min\{C,f/2^\nu\},\quad 
c (\min\{2^\nu,f\})^{\delta}f^{-1-\delta}\theta\le \theta'\le f^{-1}\theta.
\end{align*}
In addition, for any $k=2,3,\dots$ there exists $C_k$ such that uniformly in $\nu\in\mathbb N_0$  
\begin{align*}
0\le (-1)^{k-1}\theta^{(k)}\le C_kf^{-k}\theta.
\end{align*}
\end{lemma}
\begin{proof}
See \cite{MR4062329} for the proof. 
\end{proof}

The weight $\theta'$ conforms with the $\mathcal B^*$-norm in 
the following manner. 

\begin{lemma}\label{25072913}
There exist $c,C>0$ such that for any $\psi\in\mathcal B^*$ 
\begin{equation*}
c\|\psi\|_{\mathcal B^*}
\le \sup_{\nu\in\mathbb N_0}\|\theta'^{1/2}\psi\|
\le C\|\psi\|_{\mathcal B^*}.
\end{equation*}
\end{lemma}
\begin{proof}
The former inequality is obvious since there exists $C_1>0$ such that 
uniformly in $\nu\in \mathbb N_0$ 
\[
2^{-\nu}F_\nu\le C_1\theta',
\]
where $F_\nu$ is from \eqref{250729}. 
On the other hand, as for the latter, 
by noting the expression \eqref{eq:15.2.15.5.9} 
we can decompose and bound the norm for any $\nu\in\mathbb N_0$ as 
\begin{align*}
\|\theta'^{1/2}\psi\|
&\le 
\sum_{\mu=0}^{\nu-1}\|F_\mu\theta'^{1/2}\psi\|
+
\sum_{\mu=\nu}^\infty\|F_\mu\theta'^{1/2}\psi\|
\\&\le 
\sum_{\mu=0}^{\nu-1}2^{-\nu/2}\|F_\mu\psi\|
+
\sum_{\mu=\nu}^\infty 2^{\delta\nu/2}\|F_\mu f^{-(1+\delta)/2}\psi\|
\\&\le 
\sum_{\mu=0}^{\nu-1}2^{-(\nu-\mu)/2}2^{-\mu/2}\|F_\mu\psi\|
+
\sum_{\mu=\nu}^\infty 2^{-\delta(\mu-\nu)/2}2^{-\mu/2}\|F_\mu \psi\|
\\&\le 
\|\psi\|_{\mathcal B^*}
\Biggl(
\sum_{\kappa=1}^{\nu}2^{-\kappa/2}
+
\sum_{\kappa=0}^\infty 2^{-\delta\kappa/2}
\Biggr).
\end{align*}
This in fact implies the latter bound. 
\end{proof}

\subsubsection{Commutator estimate}

Now we state and prove the main commutator inequality of the section. 
Recall the notation $\sigma$ defined right after \eqref{25100422}. 

\begin{lemma}\label{lem:14.10.4.1.17ffaa}
Let $I\subset \mathbb R\setminus [m_-,m_+]$ be a compact interval. 
Then there exist $c,C>0$ and $a,b\in\mathbb N_0$ such that uniformly in $z\in I_\pm$
and $\nu\in\mathbb N_0$
\begin{equation*}
\sigma\mathop{\mathrm{Im}}\bigl(A_\Theta(H-z)\bigr)
\ge c\theta'-C\chi_b^2\theta-\mathop{\mathrm{Re}}(\gamma (H-z))
-C(H-z)^*(H-z),
\end{equation*}
respectively, 
where $\gamma=\gamma_{z,\nu}$ is a uniformly bounded $n\times n$ matrix-valued function,
i.e., $|\gamma|\le C$ independently of $z\in  I_\pm$ and $\nu\in\mathbb N_0$.
\end{lemma}
\begin{proof}
We discuss only the case with $\sigma=+$ and the upper sign. 
Fix a compact interval $I\subset (m_+,\infty)$, and choose $a\in\mathbb N_0$ such that 
 for some $c_1>0$
\begin{equation*}
\lambda-q_0\ge c_1
\ \ \text{uniformly in }\lambda\in I\text{ and }|x|\ge 2^{a}. 
\end{equation*}
The below estimates are all uniform in $z=\lambda+\mathrm i\Gamma\in I_+$ and $\nu\in\mathbb N_0$. 
In this proof for notational simplicity we let  
\begin{align*}
Q=f^{-1-\rho}\theta
+p_jf^{-1-\rho}\theta p_j
+(H-z)^*(H-z).
\end{align*}

Now we compute and bound the left-hand side of the asserted inequality.
By Lemmas~\ref{25071822} and \ref{lem:15.2.18.14.28} and the Cauchy--Schwarz inequality
we first have 
\begin{align}
\begin{split}
\mathop{\mathrm{Im}}\bigl(A_\Theta(H-z)\bigr)
&=
\mathop{\mathrm{Re}}\bigl(\alpha_f\Theta' p_f\bigr)
+ \mathop{\mathrm{Re}}\bigl(\alpha_jf^{-1}\Theta\ell_{jk}p_k\bigr)
- \bigl(\partial_f(q_0+q_1)\bigr)\Theta
\\&\phantom{{}={}}{}
+ (\Delta f)q_2\Theta 
+ q_2\Theta' 
- 2\mathop{\mathrm{Im}}(q_2\Theta p_f)
-\Gamma A_\Theta
\\&
\ge 
\mathop{\mathrm{Re}}\bigl(\alpha_f\bar\chi_a\theta' p_f\bigr)
+ \mathop{\mathrm{Re}}\bigl(\alpha_jf^{-1}\Theta\ell_{jk}p_k\bigr)
-\Gamma A_\Theta
-C_1Q
.
\end{split}\label{eq:14.9.4.18.22ffaa}
\end{align} 
The first and second terms on the right-hand side 
of \eqref{eq:14.9.4.18.22ffaa} combine,
similarly to \eqref{2507182240b}, as 
\begin{align}
\begin{split}
&
\mathop{\mathrm{Re}}\bigl(\alpha_f\bar\chi_a\theta' p_f\bigr)
+ \mathop{\mathrm{Re}}\bigl(\alpha_jf^{-1}\Theta\ell_{jk}p_k\bigr)
\\&
=
\tfrac12(\lambda-q)\bar\chi_a(\theta'+f^{-1}\theta) 
+\tfrac12\mathop{\mathrm{Re}}\bigl(\bar\chi_a(\theta'+f^{-1}\theta) (H-z)\bigr)
\\&\phantom{={}}{}{}
-\tfrac12\mathop{\mathrm{Re}}\bigl(\bar\chi_a(f^{-1}\theta-\theta')\alpha_jh_{jk}p_k\bigr)
\\&
\ge 
\tfrac12(\lambda-q_0-q_1)\bar\chi_a(\theta'+f^{-1}\theta)
- \tfrac14(\lambda-q_0)\bar\chi_a(f^{-1}\theta-\theta')
\\&\phantom{={}}{}
- \tfrac14p_ih_{il} \alpha_l(\lambda-q_0)^{-1}\bar\chi_a(f^{-1}\theta-\theta')\alpha_jh_{jk}p_k
-C_2Q
,
\end{split}
\label{2507182240bd}
\end{align}
with $h_{jk}$ being from \eqref{25072714}.
To bound the third term on the right-hand side of \eqref{2507182240bd}, 
we set 
\begin{equation*}
\Xi=(\lambda-q_0)^{-1}\bar\chi_a(f^{-1}\theta-\theta'),\quad 
\widetilde\Xi=(\lambda-\widetilde q_0)^{-1}\bar\chi_a(f^{-1}\theta-\theta')
\end{equation*}
for short, and proceed similarly to \eqref{250726}, \eqref{250726f} and \eqref{25072622} as 
\begin{align}
\begin{split}
&
- \tfrac14p_ih_{il} \alpha_l(\lambda-q_0)^{-1}\bar\chi_a(f^{-1}\theta-\theta')\alpha_jh_{jk}p_k
\\&
=
- \tfrac14\mathop{\mathrm{Re}}
\bigl(\widetilde\Xi h_{ij} h_{jk}p_ip_k\bigr)
- \tfrac14\mathop{\mathrm{Im}}
\bigl(\alpha_l(\partial_ih_{il} \Xi h_{jk})\alpha_jp_k\bigr)
\\&
\ge 
- \tfrac14\mathop{\mathrm{Re}}
\bigl(\widetilde\Xi \delta_{ik}p_ip_k\bigr)
-C_3Q
\\&
=
- \tfrac14p_i\alpha_i\Xi \alpha_kp_k
+ \tfrac14\mathop{\mathrm{Im}}\bigl(\alpha_i (\partial_i\Xi) \alpha_kp_k\bigr)
-C_3Q
\\&
\ge 
- \tfrac14(\lambda-q_0-2q_1)\bar\chi_a(f^{-1}\theta-\theta')
-C_4Q
.
\end{split}
\label{250726d}
\end{align}
Hence by \eqref{2507182240bd} and \eqref{250726d}
\begin{align}
\begin{split}
\mathop{\mathrm{Re}}\bigl(\alpha_f\bar\chi_a\theta' p_f\bigr)
+ \mathop{\mathrm{Re}}\bigl(\alpha_jf^{-1}\Theta\ell_{jk}p_k\bigr)
&
\ge 
(\lambda-q_0-q_1)\bar\chi_a\theta' 
-C_5Q
\\&
\ge 
(c_1-q_1)\theta'
-C_6Q
.
\end{split}\label{eq:14.9.4.18.22ffaad}
\end{align} 
As for the third term on the right-hand side of \eqref{eq:14.9.4.18.22ffaa}, 
we use the Cauchy--Schwarz inequality and Lemma~\ref{2507172236} with $\Xi=1$ 
to deduce 
\begin{align}
\begin{split}
-\Gamma A_\Theta
&
\ge 
-C_7\Gamma p_j p_j
-C_7\Gamma
\\&
\ge 
-C_7\Gamma \bigl(p_j p_j
-(z-q)^*(z-q)
\bigr)
-C_8\Gamma
\\&
=
-2C_7\Gamma\mathop{\mathrm{Re}}\bigl((z-q)^*(H-z)\bigr)
-C_7\Gamma(H-z)^*(H-z)
\\&\phantom{={}}{}
+C_8 \mathop{\mathrm{Im}}(H-z)
\\&
\ge 
-\mathop{\mathrm{Re}}(\gamma(H-z))
-C_9Q
,
\end{split}
\label{250728}
\end{align}
where we have set 
\[
\gamma=
2C_7\Gamma(z-q)^*
+\mathrm iC_8
\]
Therefore by \eqref{eq:14.9.4.18.22ffaa}, \eqref{eq:14.9.4.18.22ffaad} and \eqref{250728}
it follows that 
\begin{equation}
\mathop{\mathrm{Im}}\bigl(A_\Theta(H-z)\bigr)
\ge 
(c_1-q_1)\theta' 
-\mathop{\mathrm{Re}}(\gamma(H-z))
-C_{9}Q
,
\label{eq:14.9.4.18.22ffaadd}
\end{equation}

Finally we can show by Lemmas~\ref{2507172236} and \ref{lem:15.2.18.14.28} that 
\begin{align}
\begin{split}
Q
&\le 
C_{11}f^{-1-\rho}\theta
+C_{11}(H-z)^*(H-z)
\\&
\le 
C_{12}f^{-(\rho-\delta)}\theta'
+C_{11}(H-z)^*(H-z)
,
\end{split}
\label{eq:14.10.3.2.16ffaa}
\end{align}
which, combined with \eqref{eq:14.9.4.18.22ffaadd}, implies 
\begin{align*}
\mathop{\mathrm{Im}}\bigl(A_\Theta(H-z)\bigr)
&\ge 
\bigl(c_1-q_1-C_{13}f^{-(\rho-\delta)}\bigr)\theta' 
-\mathop{\mathrm{Re}}(\gamma(H-z))
\\&\phantom{={}}{}
+C_{13}(H-z)^*(H-z)
.
\end{align*}
Thus by letting $b\in\mathbb N_0$ be large enough we obtain the assertion. 
\end{proof}

\subsection{Contradiction to Rellich's theorem}\label{subsec:15.2.14.14.41}

Here we provide the proof of the LAP bounds. 

\begin{proof}[Proof of Theorem~\ref{thm:12.7.2.7.9}]

Let $I\subset \mathbb R\setminus[m_-,m_+]$ be a compact interval,
and fix $a,b\in\mathbb N_0$ as the assertion of Lemma~\ref{lem:14.10.4.1.17ffaa},

\smallskip
\noindent
\textit{Step I.}\ 
We first prove that there exists $C_1>0$ such that 
for any $\phi=R(z)\psi$ with $z\in I_\pm$ and $\psi\in\mathcal B$ 
\begin{equation}
\|p_1\phi\|_{\mathcal B^*}
+\dots
+\|p_d\phi\|_{\mathcal B^*}
\le 
C_1 \bigl(\|\phi\|_{\mathcal B^*}
+ \|\psi\|_{\mathcal B^*}\bigr)
.
\label{25072914}
\end{equation}
In fact, by Lemma~\ref{2507172236} with $\Xi=\theta'$ we have 
\begin{align*}
p_j\theta' p_j
&
=
(z-q)^*(z-q)\theta'
-\mathop{\mathrm{Im}}\bigl(\alpha_j(\partial_j\theta') (z-q)\bigr)
+\tfrac12(\Delta \theta')
\\&\phantom{={}}{}
+\mathop{\mathrm{Re}}
\bigl(\bigl(2(z-q)^*\theta'+\mathrm i\alpha_j(\partial_j\theta')\bigr)(H-z)\bigr)
+(H-z)^*\theta'(H-z)
,
\end{align*}
and, if we take the expectation of this inequality on the state $\phi=R(z)\psi$ and use
the Cauchy--Schwarz inequality, it follows that 
\begin{align*}
\|\theta'^{1/2} p_1\phi\|^2
+\dots
+\|\theta'^{1/2} p_d\phi\|^2
&\le 
C_2\bigl(\|\theta'^{1/2}\phi\|^2
+\|\theta'^{1/2}\psi\|^2\bigr)
.
\end{align*}
Hence by Lemma~\ref{25072913} we obtain the claim \eqref{25072914}.

\smallskip
\noindent
\textit{Step II.}\ 
Next we claim that there exist $C_3, C_4>0$ such that 
uniformly in $\nu\in\mathbb N_0$ and 
$\phi=R(z)\psi$ with $z\in I_\pm$ and $\psi\in\mathcal B$
\begin{align}
\|\theta'{}^{1/2}\phi\|^2
&\le C_3\bigl(\|\psi\|_{\mathcal B}^2
+\|\chi_b \theta^{1/2}\phi\|^2\bigr),
\label{eq:13.8.22.4.59cc}
\\
\|\phi\|_{\mathcal B^*}^2
&\le C_4\bigl(\|\psi\|_{\mathcal B}^2
+\|\chi_b \phi\|^2\bigr).
\label{eq:13.8.22.4.59ccc}
\end{align}
In fact, take the expectation of the inequality 
from Lemma~\ref{lem:14.10.4.1.17ffaa} on the state $\phi=R(z)\psi$,
and use the $\mathcal B$-$\mathcal B^*$ duality, 
the Cauchy--Schwarz inequality and \eqref{25072914},
and then we can immediately deduce \eqref{eq:13.8.22.4.59cc}.
In addition, \eqref{eq:13.8.22.4.59ccc} directly follows 
from \eqref{eq:13.8.22.4.59cc} and Lemmas~\ref{lem:15.2.18.14.28} and \ref{25072913}.

\smallskip
\noindent
\textit{Step III.}
Finally we prove the assertion. By \eqref{25072914} it suffices to show 
that uniformly in $\phi=R(z)\psi$ with $z\in I_\pm$ and $\psi\in\mathcal B$
\begin{align}
\|\phi\|_{\mathcal B^*}\le C_5\|\psi\|_{\mathcal B}.
\label{eq:15.2.15.5.2}
\end{align}
We assume that this is false, 
and choose $z_k\in I_\pm$ and $\psi_k\in \mathcal B$ such that 
\begin{align}
\lim_{k\to\infty}\|\psi_k\|_{\mathcal B}=0,\quad
\|\phi_k\|_{\mathcal B^*}=1;\quad
\phi_k=R(z_k)\psi_k.
\label{eq:15.2.15.5.51}
\end{align}

We first note that we may let $(z_k)_{k\in\mathbb N}$ be convergent to a point 
in $I\subset \mathbb R$, i.e., for some $\lambda\in I$
\begin{align}
\lambda=\lim_{k\to\infty}z_k\in I.
\label{eq:15.2.15.5.52}
\end{align}
In fact, 
since 
\begin{align*}
\|\phi_k\|_{\mathcal B^*}
\le 
\|\phi_k\|
\le \|R(z_k)\|_{\mathcal L(\mathcal H)}\|\psi_k\|
\le C_6|\mathop{\mathrm{Im}}z_k|^{-1}\|\psi_k\|_{\mathcal B}
,
\end{align*}
we have $\mathop{\mathrm{Im}}z_k\to 0$ as $k\to\infty$ 
due to \eqref{eq:15.2.15.5.51}. 
Then by compactness of $I$, we can choose a 
convergent subsequence of $(\mathop{\mathrm{Re}}z_k)_{k\in\mathbb N}$.

Next, we show that for any $s>1/2$ by choosing a further subsequence 
we may let $(\phi_k)_{k\in\mathbb N}$ be convergent in $L^2_{-s}$,
i.e., for some $\phi\in L^2_{-s}$
\begin{align}
\phi=\lim_{k\to\infty}\phi_k\in L^2_{-s}
.
\label{eq:15.2.15.5.53}
\end{align}
In fact, choose any $s'\in (1/2,s)$, and then 
by \eqref{eq:15.2.15.5.51} and \eqref{25072914}
the sequence $(f^{-s'}\phi_k)_{k\in\mathbb N}$ 
is uniformly bounded in $H^1(\mathbb R^d;\mathbb C^n)$. 
Since the multiplication operator 
\[
f^{-(s-s')}\colon H^1(\mathbb R^d;\mathbb C^n)\to L^2(\mathbb R^d;\mathbb C^n)
\]
is compact, we obtain the claim \eqref{eq:15.2.15.5.53}.

By \eqref{eq:15.2.15.5.51}, \eqref{eq:15.2.15.5.52} and \eqref{eq:15.2.15.5.53}
it follows that 
\begin{align}
(H-\lambda)\phi=0 \text{ in the distributional sense}
.
\label{eq:15.2.15.14.32}
\end{align}
Furthermore, we can verify $\phi\in \mathcal B^*_0$, 
so that $\phi=0$ by Theorem~\ref{250717}.
In fact, by \eqref{eq:13.8.22.4.59cc} we have 
\[
\|\theta'{}^{1/2}\phi_k\|^2
\le C_3\bigl(\|\psi_k\|_{\mathcal B}^2
+\|\chi_b \theta^{1/2}\phi_k\|^2\bigr).
\]
If we let $s=(1+\delta)/2$ and take the limit $k\to\infty$,
we obtain by \eqref{eq:15.2.15.5.53}, \eqref{eq:15.2.15.5.51} and Lemma~\ref{lem:15.2.18.14.28} that 
for any $\nu\in\mathbb N_0$
\begin{align}
\begin{split}
\|\theta'{}^{1/2}\phi\|
&\le C_7\|\chi_b\theta^{1/2}\phi\|
\le C_72^{-\nu/2} \|\chi_bf^{1/2}\phi\|
.
\end{split}
\label{eq:15.2.15.12.9}
\end{align}
Letting $\nu\to\infty$ in \eqref{eq:15.2.15.12.9},
we obtain $\phi\in \mathcal B^*_0$, and then conclude $\phi=0$.

However, this is a contradiction, because by \eqref{eq:13.8.22.4.59ccc} we have 
\begin{align*}
1=\|\phi_k\|_{\mathcal B^*}^2\le C_4\bigl(\|\psi_k\|_{\mathcal B}^2+\|\chi_b\phi_k\|^2\big),
\end{align*}
and, as $k\to \infty$, the right-hand side converges to $0$ by \eqref{eq:15.2.15.5.51} and 
\eqref{eq:15.2.15.5.53} and $\phi=0$.
Therefore \eqref{eq:15.2.15.5.2} holds true.
\end{proof}

\section{Proof of radiation condition bounds, and applications}

\subsection{Commutator method for radiation condition bounds}

Here we prove Theorem~\ref{25081317}. 
The proof, again, depends on commutator arguments, however, it is fairly different from Ito--Skibsted~\cite{MR4062329}.  
Theorem~\ref{25081317} is almost trivial from the following bounds.

\begin{proposition}\label{25081322b}
Suppose Assumption~\ref{2571017b},
let $I\subset \mathbb R\setminus[m_-,m_+]$ be a compact interval, and let $\kappa\in (0,1/2)$. 
Then there exists $C>0$ such that for any  $z\in I_\pm$ 
\begin{align*}
&
\pi_{\mp}^2f^{-1+2\kappa}
+(p_f-\alpha_f(z-q_0))^*f^{-1+2\kappa}(p_f-\alpha_f(z-q_0))
\\&
\le 
Cf^{-1-\rho+2\kappa}
+C(H-z)^*f^{1+2\kappa}(H-z)
,
\end{align*}
respectively. 
\end{proposition}

We split the proof of Proposition~\ref{25081322b} into two steps,
and for the first step we use a commutator method. 
The weight function of the section is given by 
\[
\Theta=\Theta^\kappa_a=\bar\chi_af^{2\kappa}
\]
with parameters $a\in\mathbb N$ and $\kappa> 0$. 
Recall the notation $\sigma$ defined right after \eqref{25100422}. 

\begin{lemma}\label{25100513}
Suppose Assumption~\ref{2571017b},
let $I\subset \mathbb R\setminus[m_-,m_+]$ be a compact interval, and let $\kappa> 0$. 
There exist $c,C>0$ and $a\in\mathbb N$ such that for any $z\in I_\pm$ 
\begin{align*}
&\mathop{\mathrm{Im}}\bigl((\sigma p_f\mp(z-q_0))^*\Theta (H-z)\bigr)
\\&\ge 
c\pi_{\mp}^2f^{-1+2\kappa}
-\sigma \bigl(\tfrac12-\kappa\bigr)\mathop{\mathrm{Re}}\bigl(f^{-1}\Theta(\alpha_fp_f-z+q_0)\bigr)
\\&\phantom{{}={}}{}
-Cf^{-1-\rho+2\kappa}
-C(H-z)^*f^{1+2\kappa}(H-z)
,
\end{align*}
respectively. In particular, it follows that for any $\epsilon>0$ there exist $c',C'>0$ such that for any $z\in I_\pm$ 
\begin{align*}
&
c'\pi_{\mp}^2f^{-1+2\kappa}
-\sigma \bigl(\tfrac12-\kappa\bigr)\mathop{\mathrm{Re}}\bigl(f^{-1}\Theta(\alpha_fp_f-z+q_0)\bigr)
\\&\le 
C'f^{-1-\rho+2\kappa}
+C'(H-z)^*f^{1+2\kappa}(H-z)
\\&\phantom{{}={}}{}
-\epsilon(p_f-\alpha_f(z-q_0))^*f^{-1+2\kappa}(p_f-\alpha_f(z-q_0))
,
\end{align*}
respectively. 
\end{lemma}

\begin{proof}
The latter inequality follows from the former, since by Assumption~\ref{2571017b} 
and the Cauchy--Schwarz inequality 
\begin{align*}
&
\mathop{\mathrm{Im}}\bigl((\sigma p_f\mp(z-q_0))^*\Theta (H-z)\bigr)
\\&
=
\sigma \mathop{\mathrm{Im}}\bigl((p_f-\alpha_f(z-q_0))^*\Theta (H-z)\bigr)
\mp2\mathop{\mathrm{Im}}\bigl(\pi_\mp (z^*-q_0)\Theta (H-z)\bigr)
\\&\phantom{{}={}}{}
+\sigma\mathop{\mathrm{Im}}\bigl(\alpha_f (q_0-\widetilde q_0)\Theta (H-z)\bigr)
\\&
\le 
\epsilon(p_f-\alpha_f(z-q_0))^*f^{-1+2\kappa}(p_f-\alpha_f(z-q_0))
+\tfrac12 c'\pi_{\mp}^2f^{-1+2\kappa}
\\&\phantom{{}={}}{}
+C_1f^{-1-\rho+2\kappa}
+C_1(H-z)^*f^{1+2\kappa}(H-z)
.
\end{align*}
Hence it suffices to show the former inequality of the assertion. 
Fix any compact interval $I\subset \mathbb R\setminus[m_-,m_+]$ and $\kappa>0$. 
Choose $a\in\mathbb N$ such that 
\[
\sigma(\lambda-q_0)\ge c_1>0\ \ \text{uniformly in }\lambda \in I\text{ and }|x|\ge 2^a.
\]
The below estimates are uniform in $z=\lambda+\mathrm i\Gamma\in I_\pm$,
and we set for short 
\[
Q=f^{-1-\rho+2\kappa}+p_jf^{-1-\rho+2\kappa}p_j+(H-z)^*f^{1+2\kappa}(H-z)
.
\]
Note that by Lemma~\ref{2507172236} and the Cauchy--Schwarz inequality we can see that 
\begin{equation*}
Q\ge -C_2f^{-1-\rho+2\kappa}-C_2(H-z)^*f^{1+2\kappa}(H-z),
\end{equation*}
and thus $Q$ is in fact negligible. 

Now we start to compute the left-hand side of the former assertion. 
We expand it, and use \eqref{2507191540}, \eqref{25072714}, Assumption~\ref{2571017b}, the Cauchy--Schwarz inequality 
and Lemma~\ref{25071822} to have 
\begin{align}
\begin{split}
&
\mathop{\mathrm{Im}}\bigl((\sigma p_f\mp(z-q_0))^*\Theta (H-z)\bigr)
\\&
=
\sigma\mathop{\mathrm{Im}}\bigl(p_f \Theta (H-\lambda)\bigr)
\mp\tfrac12\mathop{\mathrm{Im}}\bigl((2\lambda-q_0-\widetilde q_0)\Theta H\bigr)
-\sigma \Gamma\mathop{\mathrm{Re}}(p_f^*\Theta )
\\&\phantom{{}={}}{}
\pm\tfrac12\Gamma(2\lambda-q_0-\widetilde q_0)\Theta
\pm\tfrac12\mathop{\mathrm{Im}}\bigl((q_0-\widetilde q_0)\Theta (H-z)\bigr)
\pm\Gamma\mathop{\mathrm{Re}}(\Theta (H-z))
\\&
\ge 
\tfrac12\sigma\mathop{\mathrm{Im}}(A_\Theta H)
-\tfrac12\sigma\mathop{\mathrm{Re}}\bigl(((\Delta f)\Theta+\Theta') (H-z)\bigr)
\\&\phantom{{}={}}{}
\mp\tfrac14(2\lambda\alpha_f-\alpha_fq_0-q_0\alpha_f)\Theta' 
-\sigma \Gamma\mathop{\mathrm{Re}}\bigl(\alpha_f\Theta \alpha_jp_j\bigr)
\\&\phantom{{}={}}{}
+\sigma \Gamma\mathop{\mathrm{Re}}\bigl(\alpha_f\Theta \alpha_j\ell_{jk}p_k\bigr)
\pm\tfrac12\Gamma(2\lambda-q_0-\widetilde q_0)\Theta
\mp C_3\sigma\Gamma f^{-\rho}\Theta
-C_3Q
\\&
\ge 
\tfrac12\sigma\mathop{\mathrm{Re}}(\alpha_f\Theta'p_f)
+\tfrac12\sigma\mathop{\mathrm{Re}}(\alpha_jf^{-1}\Theta\ell_{jk}p_k\bigr)
\mp\tfrac14(2\lambda\alpha_f-\alpha_fq_0-q_0\alpha_f)\Theta' 
\\&\phantom{{}={}}{}
-\tfrac12\sigma \Gamma\Theta (2\lambda\alpha_f-\alpha_fq_0-q_0\alpha_f)
+\sigma \Gamma\mathop{\mathrm{Re}}\bigl(\alpha_l\alpha_j\Theta (\partial_lf)\ell_{jk}p_k\bigr)
\\&\phantom{{}={}}{}
\pm\tfrac12\Gamma(2\lambda-q_0-\widetilde q_0)\Theta
\mp C_4\sigma\Gamma f^{-\rho}\Theta
-C_4Q
.
\end{split}
\label{251005}
\end{align}
We combine the first to the third terms of \eqref{251005} 
by using \eqref{2507191540}, Assumption~\ref{2571017b} and the Cauchy--Schwarz inequality as 
\begin{align}
\begin{split}
&
\tfrac12\sigma\mathop{\mathrm{Re}}(\alpha_f\Theta'p_f)
+\tfrac12\sigma\mathop{\mathrm{Re}}(\alpha_jf^{-1}\Theta\ell_{jk}p_k\bigr)
\mp\tfrac14(2\lambda\alpha_f-\alpha_fq_0-q_0\alpha_f)\Theta' 
\\&
\ge 
2\sigma \kappa \pi_{\mp}(\lambda -q_0)\pi_{\mp}f^{-1}\Theta 
-\sigma\bigl(\tfrac12-\kappa\bigr)\mathop{\mathrm{Re}}(f^{-1}\Theta (\alpha_fp_f-\lambda+q_0)\bigr)
\\&\phantom{{}={}}{}
+\tfrac12\sigma\mathop{\mathrm{Re}}(f^{-1}\Theta (H-z)\bigr)
-\tfrac12\sigma f^{-1}\Theta (q_1+q_2)
-C_5Q
\\&
\ge 
c_2\pi_{\mp}^2f^{-1+2\kappa}
-\sigma\bigl(\tfrac12-\kappa\bigr)\mathop{\mathrm{Re}}(f^{-1}\Theta (\alpha_fp_f-\lambda+q_0)\bigr)
-C_6Q
.
\end{split}
\label{251005b}
\end{align}
We combine the fourth and the sixth terms as 
\begin{align}
\begin{split}
-\tfrac12\sigma \Gamma\Theta (2\lambda\alpha_f-\alpha_fq_0-q_0\alpha_f)
\pm\tfrac12\Gamma(2\lambda-q_0-\widetilde q_0)\Theta
&\ge 
\pm 2\Gamma \pi_\mp(\lambda - q_0)\pi_\mp\Theta
\\&
\ge 0.
\end{split}
\label{25101313}
\end{align}
As for the seventh term of \eqref{251005}, we bound it as  
\begin{equation}
\mp C_4\sigma\Gamma f^{-\rho}\Theta 
=
\pm C_4\sigma\mathop{\mathrm{Im}}(f^{-\rho}\Theta (H-z))
\mp C_4\sigma\mathop{\mathrm{Im}}(f^{-\rho}\Theta \alpha_jp_j)
\ge 
-C_7Q
.
\label{251005bb}
\end{equation}
We compute and bound the fifth term of \eqref{251005} by using \eqref{25071914} and \eqref{251005bb} as 
\begin{align}
\begin{split}
&
\sigma \Gamma\mathop{\mathrm{Re}}\bigl(\alpha_l\alpha_j\Theta (\partial_lf)\ell_{jk}p_k\bigr)
\\&
=
\tfrac12\sigma \Gamma\mathop{\mathrm{Re}}\bigl(\alpha_l\alpha_j\Theta (\partial_lf)\ell_{jk}p_k\bigr)
+\tfrac12\sigma \Gamma\mathop{\mathrm{Re}}\bigl(\alpha_j\alpha_lp_k\Theta (\partial_lf)\ell_{jk}\bigr)
\\&
\ge 
\mp C_8\sigma \Gamma f^{-1}\Theta 
-C_8Q
\\&
\ge 
-C_9Q
.
\end{split}
\label{2510054}
\end{align}
Therefore by \eqref{251005}--\eqref{2510054} we obtain 
\begin{align}
\begin{split}
&\sigma\mathop{\mathrm{Im}}\bigl((p_f\mp(z-q_0))^*\Theta (H-z)\bigr)
\\&
\ge 
c_2\pi_{\mp}^2f^{-1+2\kappa}
-\sigma\bigl(\tfrac12-\kappa\bigr)\mathop{\mathrm{Re}}(f^{-1}\Theta (\alpha_fp_f-\lambda+q_0)\bigr)
-C_9Q.
\end{split}
\label{251013}
\end{align}
Hence we are done. 
\end{proof}

The second step for the proof of Proposition~\ref{25081322b} is the following. 
We remark that Assumption~\ref{2571017} suffices for it. 

\begin{lemma}\label{25100514}
Suppose Assumption~\ref{2571017}, let $I\subset \mathbb R\setminus[m_-,m_+]$ be a compact interval, and let 
$\kappa> 0$.
Then there exist $c,C>0$ and $a\in\mathbb N$ such that 
for any $z\in I_\pm$
\begin{align*}
\begin{split}
-\sigma\mathop{\mathrm{Re}}\bigl(f^{-1}\Theta(\alpha_fp_f-z+q_0+q_1)\bigr)
&
\ge 
c(p_f-\alpha_f(z-q_0))^*f^{-1+2\kappa}(p_f-\alpha_f(z-q_0))
\\&\phantom{{}={}}{}
-Cf^{-2-\rho+2\kappa}
-C(H-z)^*f^{\rho+2\kappa}(H-z)
,
\end{split}
\end{align*}
respectively. 
\end{lemma}
\begin{proof}
Fix a compact interval $I\subset \mathbb R\setminus[m_-,m_+]$ and $\kappa>0$,
and choose $a\in\mathbb N$ such that 
\[
\sigma(\lambda-q_0)\ge c_1>0\ \ \text{uniformly in }\lambda \in I\text{ and }|x|\ge 2^a.
\]
The below estimates are uniform in $z=\lambda+\mathrm i\Gamma\in I_\pm$.
Set for short 
\[
Q=f^{-2-\rho+2\kappa}+p_jf^{-2-\rho+2\kappa}p_j
+(H-z)^*f^{\rho+2\kappa}(H-z)
.
\]
By Lemma~\ref{2507172236} we see that $Q$ is absorbed into the 
last two terms of the asserted inequality, and thus we can regard it as a negligible error.  
In this proof we start with the first term on the right-hand side of the asserted inequality. 
Let us denote 
\[
\Xi=\sigma (\lambda-q_0)^{-1}f^{-1}\Theta
, \quad 
\widetilde \Xi=\sigma (\lambda-\widetilde q_0)^{-1}f^{-1}\Theta, 
\]
and then by the Cauchy--Schwarz inequality and Assumption~\ref{2571017} we have 
\begin{align*}
&(\alpha_fp_f-z+q_0)^*f^{-1+2\kappa}(\alpha_fp_f-z+q_0)
\\&
\le 
C_1(\alpha_fp_f-z+q_0)^*f^{-1}\Theta(\alpha_fp_f-z+q_0)
+C_1Q
\\&
=
C_1(\alpha_fp_f-z+q)^*f^{-1}\Theta(\alpha_fp_f-z+q)
\\&\phantom{{}={}}{}
-2C_1\mathop{\mathrm{Re}}\bigl((q_1+q_2)^*f^{-1}\Theta(\alpha_fp_f-z+q)\bigr)
+C_1(q_1+q_2)^2f^{-1}\Theta
+C_1Q
\\&
\le 
C_2(\alpha_fp_f-z+q)^*f^{-1}\Theta(\alpha_fp_f-z+q)
+C_2Q
\\&
\le 
C_3(\alpha_fp_f-z+q)^*\Xi(\alpha_fp_f-z+q)
+C_2Q
\\&
=
C_3\mathop{\mathrm{Re}}\bigl(p_f^*\alpha_f\Xi(\alpha_fp_f-z+q)\bigr)
-C_3\mathop{\mathrm{Re}}\bigl((z-q_0)^*\Xi(\alpha_fp_f-z+q_0+q_1)\bigr)
\\&\phantom{={}}{}
+C_3\mathop{\mathrm{Re}}\bigl((q_1+q_2)\Xi(\alpha_fp_f-z+q_0+q_1)\bigr)
-C_3\mathop{\mathrm{Re}}\bigl((z-q)^*\Xi q_2\bigr)
+C_2Q
\\&
\le 
C_3\mathop{\mathrm{Re}}\bigl(p_f^*\alpha_f\Xi(H-z)\bigr)
-C_3\mathop{\mathrm{Re}}\bigl(p_f^*\alpha_f\Xi \alpha_j\ell_{jk}p_k\bigr)
\\&\phantom{={}}{}
-C_3\sigma\mathop{\mathrm{Re}}\bigl(f^{-1}\Theta(\alpha_fp_f-z+q_0+q_1)\bigr)
-C_3\Gamma\mathop{\mathrm{Im}}\bigl(\Xi(\alpha_fp_f-z+q_0+q_1)\bigr)
\\&\phantom{={}}{}
+\tfrac14(\alpha_fp_f-z+q_0+q_1)^*f^{-1+2\kappa}(\alpha_fp_f-z+q_0+q_1)
+C_4Q
\\&
\le 
-C_3\mathop{\mathrm{Re}}\bigl(p_f^*\alpha_f\Xi \alpha_j\ell_{jk}p_k\bigr)
-C_3\sigma\mathop{\mathrm{Re}}\bigl(f^{-1}\Theta(\alpha_fp_f-z+q_0+q_1)\bigr)
\\&\phantom{={}}{}
-C_3\Gamma\mathop{\mathrm{Im}}\bigl(\Xi(\alpha_fp_f-z+q_0+q_1)\bigr)
\\&\phantom{={}}{}
+\tfrac12(\alpha_fp_f-z+q_0)^*f^{-1+2\kappa}(\alpha_fp_f-z+q_0)
+C_5Q
.
\end{align*}
This implies that 
\begin{align}
\begin{split}
&
(\alpha_fp_f-z+q_0)^*f^{-1+2\kappa}(\alpha_fp_f-z+q_0)
\\&
\le 
-2C_3\mathop{\mathrm{Re}}\bigl(p_f^*\alpha_f\Xi\alpha_j\ell_{jk}p_k\bigr)
-2C_3\sigma\mathop{\mathrm{Re}}\bigl(f^{-1}\Theta(\alpha_fp_f-z+q_0+q_1)\bigr)
\\&\phantom{={}}{}
-2C_3\Gamma\mathop{\mathrm{Im}}\bigl(\Xi(\alpha_fp_f-z+q_0+q_1)\bigr)
+C_6Q
.
\end{split}
\label{250813}
\end{align}

We further compute and bound the right-hand side of \eqref{250813}.
As for the first term, we can proceed, omitting $C_3$ and using \eqref{25071914},  
\eqref{eq:15091112} and the Cauchy--Schwarz 
inequality, as  
\begin{align}
\begin{split}
&
-2\mathop{\mathrm{Re}}\bigl(p_f^*\alpha_f\Xi\alpha_j\ell_{jk}p_k\bigr)
\\
&=
2p_f^*\alpha_f\Xi\alpha_fp_f
-2\mathop{\mathrm{Re}}\bigl(p_i(\partial_if)(\partial_lf)\alpha_l\Xi\alpha_jp_j\bigr)
\\
&=
2p_f^*\widetilde \Xi p_f
-\mathop{\mathrm{Re}}\bigl(\alpha_l\alpha_j p_i(\partial_if)(\partial_lf)\widetilde \Xi p_j\bigr)
-\mathop{\mathrm{Re}}\bigl(\alpha_j\alpha_lp_j(\partial_if)(\partial_lf)\widetilde \Xi p_i\bigr)
\\
&=
2p_f^*\widetilde \Xi p_f
-\mathop{\mathrm{Re}}\bigl(\alpha_l\alpha_j p_i(\partial_if)(\partial_lf)\widetilde \Xi p_j\bigr)
-\mathop{\mathrm{Re}}\bigl(\alpha_j\alpha_lp_i(\partial_if)(\partial_lf)\widetilde \Xi p_j\bigr)
\\&\phantom{={}}{}
+\mathop{\mathrm{Im}}\bigl(\alpha_j\alpha_l(\partial_i(\partial_if)(\partial_lf)\widetilde \Xi) p_j\bigr)
-\mathop{\mathrm{Im}}\bigl(\alpha_j\alpha_l(\partial_j(\partial_if)(\partial_lf)\widetilde \Xi) p_i\bigr)
\\
&=
\mathop{\mathrm{Im}}\bigl(\alpha_j\alpha_l(\Delta f)(\partial_lf)\widetilde \Xi p_j\bigr)
+\mathop{\mathrm{Im}}\bigl(\alpha_j\alpha_l(\partial_lf)(\partial_f\widetilde \Xi) p_j\bigr)
\\&\phantom{={}}{}
-\mathop{\mathrm{Im}}\bigl(\alpha_j\alpha_l(\partial_j\partial_if)(\partial_lf)\widetilde \Xi p_i\bigr)
-\mathop{\mathrm{Im}}\bigl(\alpha_j\alpha_l(\partial_if)(\partial_j\partial_lf)\widetilde \Xi p_i\bigr)
\\&\phantom{={}}{}
-\mathop{\mathrm{Im}}\bigl(\alpha_j\alpha_l(\partial_if)(\partial_lf)(\partial_j\widetilde \Xi) p_i\bigr)
\\
&=
2\mathop{\mathrm{Im}}\bigl((\Delta f)\widetilde \Xi p_f\bigr)
-\mathop{\mathrm{Im}}\bigl(\alpha_f(\Delta f)\Xi \alpha_jp_j\bigr)
+2\mathop{\mathrm{Im}}((\partial_f\widetilde \Xi) p_f)
\\&\phantom{={}}{}
-\mathop{\mathrm{Im}}\bigl(\alpha_f(\partial_f\Xi) \alpha_jp_j\bigr)
+\mathop{\mathrm{Im}}\bigl(\alpha_ff^{-1}\Xi \alpha_j\ell_{ji}p_i\bigr)
-\mathop{\mathrm{Im}}\bigl((\Delta f)\widetilde \Xi p_f\bigr)
\\&\phantom{={}}{}
-2\mathop{\mathrm{Im}}((\partial_f\widetilde \Xi) p_f)
+\mathop{\mathrm{Im}}(\alpha_f\alpha_j(\partial_j\widetilde \Xi) p_f)
\\
&=
-\mathop{\mathrm{Im}}\bigl(\alpha_f(f^{-1}-(\Delta f))\Xi (\alpha_fp_f-z+q_0+q_1+q_2)\bigr)
\\&\phantom{={}}{}
+\mathop{\mathrm{Im}}\bigl(\alpha_f(f^{-1}-(\Delta f))\Xi (H-z)\bigr)
+\mathop{\mathrm{Im}}(\alpha_f\alpha_j\ell_{jk}(\partial_k\widetilde \Xi) p_f)
\\&
\le 
\tfrac14C_3^{-1}(\alpha_fp_f-z+q_0)^*f^{-1+2\kappa}(\alpha_fp_f-z+q_0)
+C_7Q
.
\end{split}
\label{250826223}
\end{align}
On the other hand, the third term on the right-hand side of \eqref{250813} is bounded by 
using the Cauchy--Schwarz inequality and \eqref{251005bb} twice as 
\begin{align}
\begin{split}
&-2C_3\Gamma\mathop{\mathrm{Im}}\bigl(\Xi(\alpha_fp_f-z+q_0+q_1)\bigr)
\\&
\le 
\tfrac18(\alpha_fp_f-z+q_0+q_1)^*f^{-1+2\kappa}(\alpha_fp_f-z+q_0+q_1)
+C_8\Gamma^2f^{-1+2\kappa}
\\&
\le 
\tfrac14(\alpha_fp_f-z+q_0)^*f^{-1+2\kappa}(\alpha_fp_f-z+q_0)
+C_9Q
. 
\end{split}
\label{25082923}
\end{align}
Therefore by \eqref{250813}, \eqref{250826223} and \eqref{25082923} we obtain the assertion. 
\end{proof}

\begin{proof}[Proof of Proposition~\ref{25081322b}]
The assertion is clear from Lemmas~\ref{25100513} and \ref{25100514}. 
\end{proof}

\begin{proof}[Proof of Theorem~\ref{25081317}]
The assertion almost trivially follows from Proposition~\ref{25081322b} and Theorem~\ref{thm:12.7.2.7.9}. 
We only remark that, 
in order to verify that all weighted norms finite, 
we use the boundedness of $R(z)$ as $L^2_s\to L^2_s$ 
for any $z\in\mathbb C\setminus\mathbb R$ and $s\in\mathbb R$.  
Thus we are done.
\end{proof}

\subsection{Applications}

At last we verify Corollaries~\ref{cor:Limiting-Absorption-Principle-Stark}, 
\ref{cor:RC-bound-real} and \ref{cor:Sommerfeld-unique-result}, 
as applications of Theorems~\ref{250717}, \ref{thm:12.7.2.7.9} and \ref{25081317}.

\subsubsection{Proof of LAP}

\begin{proof}[Proof of Corollary~\ref{cor:Limiting-Absorption-Principle-Stark}]
Take any $I$, $s$ and $\epsilon$ as in the assertion,
and set $s'=s-\epsilon$. 
Let us discuss on $I_+$ since the arguments for $I_-$ are completely the same. 
For any $a\in\mathbb N_0$ and $z,w\in I_+$ we decompose 
\begin{equation}\label{eq:1812121815}
\begin{split}
R(z) - R(w) 
&= 
\bigl(\chi_aR(z)\chi_a - \chi_aR(w)\chi_a \bigr)
+ \bigl( R(z) - \chi_aR(z)\chi_a \bigr) 
\\&\phantom{{}={}}{}
- \bigl( R(w) - \chi_aR(w)\chi_a \bigr),
\end{split}
\end{equation}
see \eqref{eq:11.7.11.5.14} for $\chi_a$, and $\bar\chi_a$ appearing below. 
By Theorem~\ref{thm:12.7.2.7.9} 
we can bound the second term on the right-hand side of \eqref{eq:1812121815} 
uniformly in $a\in\mathbb N_0$ and $z\in I_+$ as 
\begin{equation} \label{eq:1812121817}
\begin{split}
\| R(z) - \chi_aR(z)\chi_a \|_{\mathcal L(L^2_s, L^2_{-s})} 
&\le 
\| f^{-s}R(z)\bar\chi_af^{-s} \|_{\mathcal L(\mathcal H)} 
\\&\phantom{{}={}}{}
+ \| f^{-s}\bar\chi_aR(z)\chi_af^{-s} \|_{\mathcal L(\mathcal H)} 
\\&\le 
C_12^{-a(s-s')}
\\&
=C_12^{-a\epsilon}.
\end{split}
\end{equation}
Similarly for the third term of \eqref{eq:1812121815}, 
we have uniformly in $a\in\mathbb N_0$ and $w\in I_+$  
\begin{equation} \label{eq:1812121819}
\| R(w) - \chi_aR(w)\chi_a \|_{\mathcal L(L^2_s, L^2_{-s})} 
\le 
C_12^{-a\epsilon }.
\end{equation}
As for the first term of \eqref{eq:1812121815}, we rewrite it as 
\begin{align*}
\chi_aR(z)\chi_a - \chi_aR(w)\chi_a 
&= 
\mathrm i\chi_aR(z)\alpha_f\chi_{a+1}'R(w)\chi_a 
+(z-w)\chi_aR(z)\chi_{a+1}R(w)\chi_a 
\\&= 
\mathrm i\chi_aR(z)\pi_{f,+}\chi_{a+1}'R(w)\chi_a 
-\mathrm i\chi_aR(z)\chi_{a+1}'\pi_{f,-}R(w)\chi_a 
\\&\phantom{{}={}}{}
+(z-w)\chi_aR(z)\pi_{f,+}\chi_{a+1}R(w)\chi_a 
\\&\phantom{{}={}}{}
+(z-w)\chi_aR(z)\chi_{a+1}\pi_{f,-}R(w)\chi_a 
.
\end{align*}
Then by Theorems~\ref{thm:12.7.2.7.9} and \ref{25081317}
uniformly in $a\in\mathbb N_0$ and $z, w \in I_+$
\begin{equation}
\|\chi_aR(z)\chi_a - \chi_aR(w)\chi_a\|_{\mathcal L(L^2_s, L^2_{-s})}
\le 
C_22^{-a\epsilon }
+C_22^{a(1-\epsilon)}|z-w|.
\label{eq:1812121827}
\end{equation}
By \eqref{eq:1812121815}, \eqref{eq:1812121817}, \eqref{eq:1812121819} 
and \eqref{eq:1812121827}, we obtain uniformly in $a\in\mathbb N_0$ and $z, w \in I_+$
\[
\| R(z) - R(w) \|_{\mathcal L(L^2_s, L^2_{-s})} 
\le 
C_32^{-\epsilon n}
+C_22^{(1-\epsilon)n}|z-w|.
\]

Now, if $|z-w|\le 1$, choose $a\in\mathbb N_0$ such that $2^a \le |z-w|^{-1} \le 2^{a+1}$, 
and then 
\begin{equation*}
\| R(z) - R(w) \|_{\mathcal L(L^2_s, L^2_{-s})} 
\le 
C_4|z-w|^\epsilon.
\end{equation*}
The same bound is trivial for $|z-w|>1$,
and hence the H\"older continuity \eqref{eq:Holder-continuity} is verified for $R(z)$. 
We can argue similarly for $p_jR(z)$.

The existence of the limits \eqref{eq:uniform-limit-z-to-lambda} follows from \eqref{eq:Holder-continuity}. 
In addition, by Theorem~\ref{thm:12.7.2.7.9}
$R_\pm(\lambda)$ and $p_jR(\lambda)$ map into $\mathcal B^*$,
and then with the density argument they extend continuously as $\mathcal B\to\mathcal B^*$. 
Hence we are done.
\end{proof}

\subsubsection{Proof of radiation condition bounds for limiting resolvents}

\begin{proof}[Proof of Corollary~\ref{cor:RC-bound-real}]
The assertion is a direct consequence of 
Theorem~\ref{25081317} and Corollary~\ref{cor:Limiting-Absorption-Principle-Stark}.
Hence we are done.
\end{proof}

\subsubsection{Proof of Sommerfeld's uniqueness}

\begin{proof}[Proof of Corollary~\ref{cor:Sommerfeld-unique-result}]
Let $\lambda\in\mathbb R\setminus[m_-,m_+]$ and $\kappa\in (0,\rho/2)$. 
First, for any $\psi\in L^2_{1/2+\kappa}$ set $\phi = R_\pm(\lambda)\psi$. 
Then  by Corollaries~\ref{cor:Limiting-Absorption-Principle-Stark} and \ref{cor:RC-bound-real} 
we can see that $\phi\in\mathcal B^*$, 
and that the conditions \ref{item:18122818} and \ref{item:18122819} hold true.
Conversely, let $\phi\in f^\kappa\mathcal B^*$ and $\psi\in f^{-\kappa}\mathcal B$,
and assume the conditions \ref{item:18122818} and \ref{item:18122820}.
If we set
\[
\Phi = \phi - R_\pm(\lambda)\psi\in\psi\in f^{\kappa}\mathcal B^*
,
\]
then by Corollaries~\ref{cor:Limiting-Absorption-Principle-Stark} 
and \ref{cor:RC-bound-real} we have 
\begin{enumerate}
\item[1${}'$.]
$(H-\lambda)\Phi=0$ in the distributional sense,
\item[3${}'$.]
$\pi_{\mp}\Phi\in f^{-\kappa}\mathcal B_0^*$,
or 
$(p_f\mp\sigma (\lambda-q_0))\Phi\in f^{-\kappa}\mathcal B_0^*$ 
\end{enumerate}
Here we note that 1$'$ and 3${}'$ imply $\pi_{\mp}\Phi\in \mathcal B_0^*$. 
To see this, it suffices to consider the case $(p_f\mp\sigma (\lambda-q_0))\Phi\in f^{-\kappa}\mathcal B_0^*$ in 3$'$. 
For any $a\in\mathbb N$ let us compute  
\begin{align*}
\mathop{\mathrm{Re}}(\alpha_f\bar\chi_a'(H-\lambda)) 
&= 
\mathop{\mathrm{Re}}\bigl(\bar\chi_a'(p_f\mp\sigma(\lambda-q_0))\bigr) 
\pm\sigma 2\mathop{\mathrm{Re}}\bigl(\pi_\mp\bar\chi_a'(\lambda-q_0)\bigr) 
\\&\phantom{{}={}}{}
+\mathop{\mathrm{Re}}\bigl(\alpha_f\alpha_j\bar\chi_a'\ell_{jk}p_k\bigr) 
+\mathop{\mathrm{Re}}\bigl(\alpha_f\bar\chi_a'(q_1+q_2)\bigr) 
\\&= 
\mathop{\mathrm{Re}}\bigl(\bar\chi_a'(p_f\mp\sigma(\lambda-q_0))\bigr) 
\pm\sigma 2\pi_\mp\bar\chi_a'(\lambda-q_0) \pi_\mp
\\&\phantom{{}={}}{}
\mp\sigma \tfrac12\bar\chi_a'(q_0-\widetilde q_0)
+\mathop{\mathrm{Re}}\bigl(\alpha_f\bar\chi_a'(q_1+q_2)\bigr), 
\end{align*}
and then, taking the expectation of the above identity on the state $\Phi$, we obtain as $a\to\infty$
\begin{align*}
0&\le 
c_1\langle\pi_\mp\Phi,\bar\chi_a'\pi_\mp\Phi\rangle 
\\&
\le 
\bigl\langle\pi_\mp\Phi,\bar\chi_a'(\lambda-q_0) \pi_\mp\Phi\bigr\rangle 
\\&
=
\mp\sigma\tfrac12\mathop{\mathrm{Re}}\bigl\langle\Phi,\chi_a'(p_f\mp\sigma(\lambda-q_0))\Phi\bigr\rangle 
+ \tfrac14\bigl\langle\Phi,\chi_a'(q_0-\widetilde q_0)\Phi\bigr\rangle
\\&\phantom{{}={}}{}
\mp\sigma\tfrac12\mathop{\mathrm{Re}}\bigl\langle\Phi,\alpha_f\chi_a'(q_1+q_2)\Phi\bigr\rangle
\\&
\to 0
.
\end{align*}
Thus we obtain $\pi_{\mp}\Phi\in \mathcal B_0^*$. 
Finally let us further verify that $\Phi \in\mathcal B_0^*$.
To see this, we compute for any  $a\in\mathbb N$ 
\[
2\mathop{\mathrm{Im}}(\chi_a(H-\lambda)) 
= 
\alpha_f\chi_a'
= 
\pm \chi_a'
\mp 2\pi_{\mp}\chi_a'
,
\]
and this and 1$'$ imply that as $a\to\infty$
\begin{equation*}
0 \le 
\langle \Phi,\bar\chi_a'\Phi \rangle
= 
2\langle \Phi,\bar\chi_a'\pi_{\mp}\Phi\rangle
\to 0.
\end{equation*}
Thus we obtain $\Phi \in\mathcal B^*_0$, and then by Theorem~\ref{250717} 
we conclude $\Phi=0$, or $\phi = R_\pm(\lambda)\psi$. 
Hence we are done.
\end{proof}

\bigskip
\noindent
\subsubsection*{Acknowledgements} 
S.A. is supported by Forefront Physics and Mathematics Program to
Drive Transformation (FoPM), a World-leading Innovative Graduate Study
(WINGS) Program, the University of Tokyo. 
K.I. is partially supported by JSPS KAKENHI Grant Number 23K03163. 
A part of this work was done while K.I. visited Aarhus University. 
He would like to express his gratitude for kind hospitality by Arne Jensen and Erik Skibsted there. 
The authors would like to sincerely thank Kyohei Itakura for pointing out a serious error
as well as many small ones  in the former versions of the paper.

\providecommand{\bysame}{\leavevmode\hbox to3em{\hrulefill}\thinspace}
\providecommand{\MR}{\relax\ifhmode\unskip\space\fi MR }
\providecommand{\MRhref}[2]{%
  \href{http://www.ams.org/mathscinet-getitem?mr=#1}{#2}
}
\providecommand{\href}[2]{#2}

\end{document}